\documentclass{article}

\usepackage{arxiv}

\usepackage[utf8]{inputenc} 
\usepackage[T1]{fontenc}    
\usepackage{hyperref}       
\usepackage{url}            
\usepackage{booktabs}       
\usepackage{amsfonts}       
\usepackage{nicefrac}       
\usepackage{microtype}      
\usepackage{lipsum}		
\usepackage{array}
\usepackage{fixltx2e}
\usepackage{stfloats}
\usepackage{cite}
\usepackage{amsmath,amssymb,amsfonts}
\usepackage{algorithmic}
\usepackage{graphicx}
\usepackage{textcomp}
\usepackage{amssymb}
\usepackage{graphicx}
\usepackage{comment}
\usepackage{extarrows}
\usepackage{multirow}
\usepackage{threeparttable}
\usepackage{booktabs}
\usepackage{comment}
\usepackage{extarrows}
\usepackage{multirow}
\usepackage{threeparttable}
\usepackage{booktabs}
\usepackage{array}
\usepackage{multirow}
\usepackage{enumerate}
\usepackage{url}
\usepackage{caption}
\usepackage{graphics}
\usepackage{amsthm}
\usepackage{epstopdf}

\title{A Comprehensive Formal Security Analysis and Revision of the Two-phase Key Exchange Primitive of TPM 2.0}

\author{
Qianying~Zhang \thanks{This paper is an extension of the conference version appearing in the Proceedings of the 8th International Conference on Trust and Trustworthy Computing (Trust'15) \cite{TPM2.0AKE}. This paper presents the details of the proof of the key exchange primitive of the TPM 2.0 specifications, gives concrete suggestions on how to revise the key exchange primitive of TPM 2.0 to make it applicable in real-world networks and achieve a
higher level of security, and rigorously analyzes the revised key exchange primitive.}\\
    College of Information Engineering \\
    Capital Normal University \\
    \texttt{qyzhang@cnu.edu.cn} \\
\And
Shijun~Zhao \thanks{(Corresponding author: Shijun Zhao)} \\
        Huawei Technologies Co., Ltd. \\
       \texttt{zqyzsj@gmail.com} \\
}

\begin{document}
\maketitle

\begin{abstract}
The Trusted Platform Module (TPM) version 2.0 provides a two-phase key exchange primitive which can be used to implement three widely-standardized authenticated key exchange protocols: the Full Unified Model, the Full MQV, and the SM2 key exchange protocols. However, vulnerabilities have been found in all of these protocols. Fortunately, it seems that the protections offered by TPM chips can mitigate these vulnerabilities. In this paper, we present a security model which captures TPM's protections on keys and protocols' computation environments and in which multiple protocols can be analyzed in a unified way. Based on the unified security model, we give the first formal security analysis of the key exchange primitive of TPM 2.0, and the analysis results show that, with the help of hardware protections of TPM chips, the key exchange primitive indeed satisfies the well-defined security property of our security model, but unfortunately under some impractical limiting conditions, which would prevent the application of the key exchange primitive in real-world networks. To make TPM 2.0 applicable to real-world networks, we present a revision of the key exchange primitive of TPM 2.0, which can be secure without the limiting conditions. We give a rigorous analysis of our revision, and the results show that our revision achieves not only the basic security property of modern AKE security models but also some further security properties.
\end{abstract}


\section{Introduction}\label{sec:intro}
The TPM, presented by the Trusted Computing Group (TCG), is widely available in modern PCs and laptops, and is used to enhance the security of modern operating systems: the Linux kernel integrity subsystem leverages the TPM to store integrity measurements of software and attest to remote entities; Windows adds many security features based on the TPM, such as trusted boot, full disk encryption (Bitlocker), device attestation, and credential protection; the Chrome OS uses the TPM to prevent firmware rollback attacks and to protect users' sensitive data and keys. Besides, TPM also presents a promising usage in the Industry 4.0 \cite{tpm-industry4.0}: it has been widely used as the hardware root of trust for the Industry 4.0 devices and can protect the trustworthiness, identity, and communication of the devices. Take the automotive use case of Industry 4.0 for example, TCG publishes TPM specifications for components of vehicle systems \cite{TPM2.0auto,TPM2.0auto-protect}: an Automotive-Rich TPM having rich capabilities can be defined for Head Unit or Gateway components of vehicle systems with powerful processing, networking and applications functionality, and an Automotive-thin TPM having fewer capabilities can be defined for Head Unit or Gateway components with limited resources.

The newest TPM specifications, TPM 2.0 \cite{TPM2.0,TPM2.0Commands}, allow each platform to choose the functions needed and the level of security required. This flexibility allows TPM 2.0 to be implemented in different types such as security chips and protected software in a trusted execution environment (TEE) \cite{tzwp}, so that TPM chips can be applied to different kinds of platforms from low-end embedded devices to PCs and even servers.
In conclusion, the TPM plays an important role in the security of both personal devices and industrial devices.

In this paper, we focus on the secure communication feature of TPM 2.0: the authenticated key exchange (AKE) functionality, which provides secure communication for devices. AKE is an important public key primitive in modern cryptography, which allows two parties to establish a shared secret session key via a public insecure communication while providing mutual authentication. To prevent active attacks, AKE protocols usually use digital signatures or message authentication codes (MAC) to explicitly authenticate the messages exchanged \cite{STS,SIGMA,TLS,JFK}. However, these authentication mechanisms incur significant overhead in both computation and communication.
To overcome the disadvantages of the explicitly AKE protocols, the implicitly AKE protocols \cite{Mat86,MQV95,KEA98,MQV03,JKL04,HMQV05,KEA+06,eCK07,CMQV08,mOT10,SM210,OAKEfull,OAKEccs} are proposed. This kind of protocols only requires basic Diffie-Hellman exchanges while providing identity authentication by combining the ephemeral keys and long-term keys during the derivation of the session key, and achieves efficiency in both computation and communication.

The TPM 2.0 specifications support three implicitly AKE protocols: the Full Unified Model (UM) \cite{um2007}, Full MQV \cite{MQV95}, and SM2 key exchange protocols \cite{SM210}. All the three protocols have been widely standardized: the Full UM and Full MQV protocols are standardized by ANSI \cite{ANSIX942,ANSIX963}, NIST \cite{NIST800-56A}, IEEE \cite{ISO15946}; the SM2 key exchange protocol is standardized by the Chinese Government State Cryptography Administration \cite{GB32918} and some industrial standards \cite{java}. TPM 2.0 describes implicitly AKE protocols as two-phase key exchange protocols because TPM 2.0 implements them in two phases. In the first phase, the TPM generates an ephemeral DH key and sends its public part to the other party. In the second phase, the TPM generates the unhashed shared secret by combining ephemeral keys and long-term keys, and then the host of the TPM uses the unhashed shared secret to derive a session key. Some works \cite{sHMQV,impHMQV,ukshsm} using the TPM to improve the efficiency or security of AKE protocols have been proposed.

We first introduce some notations used in this paper. Let $G'$ be a finite Abelian group of order $N$, $G\subseteq G'$ be a subgroup of prime order $q$. Denote by $g$ a generator of $G$, by $1_G$ the identity element, by $G\backslash1_G=G-\{1_G\}$ the set of elements of $G$ except $1_G$, and by $h=N/q$ the cofactor. We use the multiplicative notation for the group operation in $G'$. Let $u\in_{R}Z_q$ denote randomly selecting an integer $u$ between $1$ and $q-1$. Note that $G$ is an elliptic curve in this work, for all the three protocols are based on elliptic curve cryptography. Let $P.x$ denote the $x$-coordinate of point $P$. The party having $A$ as its public key will be denoted by $\hat{A}$. The Full UM, Full MQV and SM2 key exchange protocols are described in Figure \ref{fig:protocols}. We let $\lambda$ denote the security parameter, and $H_1():\{0,1\}^* \rightarrow \{0,1\}^{\lambda}$ and $H_2():\{0,1\}^* \rightarrow \{0,1\}^{\lambda}$ are cryptographic hash functions. The Full UM protocol analyzed in this paper includes the ephemeral public keys exchanged as suggested by \cite{JKL04}. The Full MQV protocol is a variant of the original MQV protocol \cite{MQV95} which does not include parties' identifiers in the session key derivation.

\subsection{Security models for AKE protocols}\label{subsec:AKE-model}
In the early days, AKE protocols are designed in an ad-hoc manner and cannot provide well-defined security properties. To deal with this problem, some rigorous security models are proposed, and by now, it has become a basic requirement for AKE protocols to achieve the security properties defined by AKE security models. The first security model is proposed by Bellare and Rogaway \cite{BR}, which is called BR model. The BR model allows the attacker to control all communications between the parties, learn secret state of some party by corrupting him, and obtain the session key of some session by revealing it. However, the BR model does not allow the attacker to get the ephemeral state of some session and the lack of simulating this attack activity prevents BR model to simulate some practical attack situations. To deal with this flaw, the CK model \cite{CKmodel} and the eCK model \cite{eCK07} simulate the leakage of ephemeral state of some session by adding the \emph{SessionStateReveal} and \emph{EphemeralKeyReveal} queries to the model respectively. The \emph{SessionStateReveal} query allows the attacker to learn the session private information while the \emph{EphemeralKeyReveal} query allows the attacker to learn the ephemeral private key -- which is stronger is still under discussion. Okamoto \cite{ot-2007} claims that \emph{EphemeralKeyReveal} is stronger than \emph{SessionStateReveal}; Cremers \cite{crem09} argues that \emph{SessionStateReveal} is stronger than \emph{EphemeralKeyReveal}; Boyd et al. \cite{boyd-incom} show that the two queries are incomparable, while Ustaoglu \cite{model-comp} shows that the functions of the two queries are essentially the same.

Since the CK model can largely provide no weaker security than the eCK model and the \emph{SessionStateReveal} query of the CK model models the high risk of the leakage of the secrets in the unprotected memory of the host better than the \emph{EphemeralKeyReveal} query of the eCK model, we build our security model based on the CK model.

\begin{figure*}
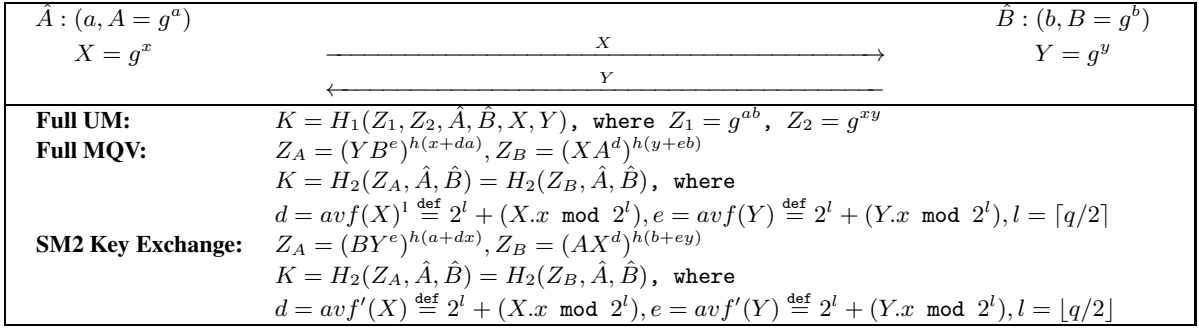

\small
  \centering{\tt
  \begin{tabular}{|>{\centering}p{70pt}>{\centering}p{250pt}p{80pt}<{\centering}|}
  \hline
  $\hat{A}:(a,A=g^a)$ & & $\hat{B}:(b,B=g^b)$\\
  $X=g^x$ &$\xrightarrow{\hspace*{100pt}X\hspace*{100pt}}$ & $Y=g^y$\\
  & $\xleftarrow{\hspace*{100pt}Y\hspace*{100pt}}$ & \\
 \hline
  \multicolumn{1}{|l}{\hspace{2mm}\textrm{\textbf{Full UM:}}} & \multicolumn{2}{l|}{$K=H_1(Z_1,Z_2,\hat{A},\hat{B},X,Y)$, where $Z_1=g^{ab}$, $Z_2=g^{xy}$}\\
  \multicolumn{1}{|l}{\hspace{2mm}\textrm{\textbf{Full MQV:}}} & \multicolumn{2}{l|}{$Z_A=(YB^e)^{h(x+da)},Z_B=(XA^d)^{h(y+eb)}$}\\
  \multicolumn{1}{|l}{\hspace{2mm}}&\multicolumn{2}{l|}{$K=H_2(Z_A,\hat{A},\hat{B})=H_2(Z_B,\hat{A},\hat{B})$, where}\\
  \multicolumn{1}{|l}{\hspace{2mm}} & \multicolumn{2}{l|}{$d=avf(X)\footnotemark[1]\overset{\text{def}}{=}2^l+(X.x$ mod $2^l),e=avf(Y)\overset{\text{def}}{=}2^l+(Y.x$ mod $2^l),l=\lceil q/2\rceil$}\\
  \multicolumn{1}{|l}{\hspace{2mm}\textrm{\textbf{SM2 Key Exchange:}}} & \multicolumn{2}{l|}{$Z_A=(BY^e)^{h(a+dx)},Z_B=(AX^d)^{h(b+ey)}$}\\
  \multicolumn{1}{|l}{\hspace{2mm}\textrm{}} & \multicolumn{2}{l|}{$K=H_2(Z_A,\hat{A},\hat{B})=H_2(Z_B,\hat{A},\hat{B})$, where}\\
  \multicolumn{1}{|l}{\hspace{2mm}} & \multicolumn{2}{l|}{$d=avf'(X)\overset{\text{def}}{=}2^l+(X.x$ mod $2^l),e=avf'(Y)\overset{\text{def}}{=}2^l+(Y.x$ mod $2^l),l=\lfloor q/2\rfloor$}\\
  \hline
  \end{tabular}}
  \caption{\label{fig:protocols}The Full UM, Full MQV, and SM2 key exchange Protocols}
\end{figure*}

\subsection{Weaknesses of the AKE protocols in TPM 2.0}\label{subsec:weak-AKE}
Unfortunately, all the three AKE protocols adopted by TPM 2.0 are not secure. We summarize their weaknesses in the following.

\footnotetext[1]{The $avf()$ and $avf'()$ functions receive an elliptic curve point and return a fix length of the least significant bits (LSB) of the x-coordinate.}

We find that the Full UM protocol is completely insecure if an attacker is able to learn the intermediate information $Z_1=g^{ab}$ of some session established by $\hat{A}$ with $\hat{B}$: the attacker transmits an ephemeral key $X'=g^{x'}$ generated by himself to party $\hat{B}$ and receives an ephemeral public key $Y'$ from $\hat{B}$, then he can compute the session key $K=H(Z_1,Y'^{x'},\hat{A},\hat{B},X',Y')$, i.e., the attacker is able to impersonate $\hat{A}$ to $\hat{B}$ indefinitely.

Kaliski presents an unknown-key share (UKS) attack \cite{mqvuks01} on the original MQV protocol: the attacker $\mathcal{M}$ interfaces with the session establishment between two honest parties $\hat{A}$ and $\hat{B}$ such that $\hat{A}$ is convinced that he is sharing a key with $\hat{B}$, but $\hat{B}$ believes that he is sharing the same session key with $\mathcal{M}$. This UKS attack requires $\mathcal{M}$ to register a specific key $C=g^c$ with the certificate authority (CA) and send a specific ephemeral public key $X'$ to $\hat{B}$. $c$ and $X'$ are so carefully computed by $\mathcal{M}$ that the session keys of sessions $(\hat{A},\hat{B},X,Y)$ and $(\hat{B},\mathcal{M},Y,X')$ are identical. The details of the UKS attack are described in \ref{app:kaliski}. Although the Full MQV protocol tries to prevent the above UKS attack by including identities in the session key derivation, we find that it still cannot achieve the security defined by modern AKE models if $\mathcal{M}$ is able to learn the unhashed shared $Z$ value: $\mathcal{M}$ performs the same steps above, learns $Z_B$ by corrupting the session $(\hat{B},\mathcal{M},Y,X')$, then $\mathcal{M}$ can compute the session key of session $(\hat{A},\hat{B},X,Y)$, i.e., the corruption of the session $(\hat{B},\mathcal{M},Y,X')$ helps $\mathcal{M}$ to compromise another session $(\hat{A},\hat{B},X,Y)$.

Xu et al. introduce two UKS attacks \cite{SM210} on the SM2 key exchange protocol in which an honest party $\hat{A}$ is coerced to share a session key with the attacker $\mathcal{M}$, but $\hat{A}$ thinks that he is sharing the key with another party $\hat{B}$. Both attacks require $\mathcal{M}$ to reveal the unhashed shared $Z_B$ of $\hat{B}$. Besides, the first attack requires $\mathcal{M}$ to register with the CA a specific key $C=Ag^e$ where $e\in_R Z_q$, and the second attack requires $\mathcal{M}$ to perform some computations using his private key after obtaining $Z_B$. The details of the two attacks are described in  \ref{app:xu}.

The above attacks show that the three AKE protocols cannot achieve the security property defined by modern AKE security models if the attacker is able to get the unhashed values. Unfortunately, this is exactly how the two-phase key exchange primitive of TPM 2.0 (denoted by $\mathsf{tpm.KE}$) implements these three AKE protocols: $Z_1$ of the Full UM protocol, the unhashed $Z$ values of the MQV and SM2 key exchange protocols are returned to the host, whose memory is vulnerable to attacks. So it seems that $\mathsf{tpm.KE}$ is not secure.

\subsection{Motivations and Contributions}\label{subsec:motivation}
Fortunately, protections provided by the TPM improve the security of $\mathsf{tpm.KE}$. First, all long-term keys are generated randomly in the TPM, so attackers cannot make the TPM to generate a specific key such as the carefully computed key $C=g^c$ required by Kaliski's UKS attack or $C=Ag^e$ required by Xu's first attack. Second, the TPM only provides well-defined functions through TPM commands \cite{TPM2.0Commands} in a black-box manner: when a TPM command is invoked, the TPM chip executes the pre-defined computation procedure and returns the computation result. The second feature prevents attackers from using a key to perform computations at will. So it seems that the above two features can prevent Kaliski's UKS attack and Xu's attacks. However, in modern cryptography, rigorous proofs of security in modern security models \cite{CKmodel,eCK07} have become a basic requirement for cryptographic protocols to be standardized and are essential tools to guarantee the soundness of cryptographic protocols. As the TPM has been widely applied on kinds of computation platforms, it is important to perform formal analysis about its security mechanisms. This leads to our first motivation:
\begin{enumerate}
\item \textit{How to build a security model which can precisely model the TPM's protections on keys and protocol computation environments and based on which we can perform rigorous security analysis of $\mathsf{tpm.KE}$?}
\end{enumerate}

Although protections provided by the TPM help the MQV and SM2 key exchange protocols to resist current UKS attacks, the $avf()$ and $avf'()$ functions used in the MQV and SM2 key exchange protocols respectively make that the two protocols cannot be proven secure. Consider such a group $G$ that the representations of its elements satisfy that the $\lceil q/2\rceil$ least significant bits (LSBs) of the representations of points' $x$-coordinate are fixed. In this case, an attacker can mount the so-called group representation attacks on the MQV and SM2 key exchange protocols, in which the attacker can impersonate $\hat{A}$ to $\hat{B}$ without knowing the private key of $\hat{A}$. The group representation attack on MQV is described in \ref{app:groupattack}, and a similar attack on the SM2 key exchange protocol can be found in \cite{zhaoSM2}. HMQV, a variant of MQV, prevents this type of attack by replacing $avf()$ with a cryptographic hash function, which enables the protocol to be proven secure in the CK model. Zhao et al. \cite{zhaoSM2} also suggest replacing the $avf'()$ function of the SM2 key exchange protocol with a cryptographic hash function. However, group representation attacks are not practical, for it is difficult to find an elliptic curve whose $\lceil q/2\rceil$ LSBs of the representations of points' $x$-coordinate are fixed. On the contrary, the outputs of the $avf()$ and $avf'()$ functions seem to range in a uniform way over all possible values. This leads to our second motivation:
\begin{enumerate}\setcounter{enumi}{1}
\item \textit{Can we give a quantitative measure of the amount of randomness (entropy) contained in the output distributions of $avf()$ and $avf'()$ on real-world elliptic curves and check whether $avf()$ and $avf'()$ provide enough entropy to prevent group representation attacks?}
\end{enumerate}

As far as we know, current modern AKE security models only consider how to formally analyze one single protocol, thus all AKE protocols proven secure in the literature are analyzed separately. However, $\mathsf{tpm.KE}$ is designed to support three implicitly AKE protocols through unified interfaces, so we cannot use current security models to analyze $\mathsf{tpm.KE}$. Besides, the design of $\mathsf{tpm.KE}$ brings the following security problem that should be considered in its analysis. Suppose an honest party $\hat{A}$ tries to establish a secure channel with $\hat{B}$ through MQV and an attacker controls a long-term key of $\hat{B}$'s TPM, but the type of the key is SM2. Then the question is whether the session key of $\hat{A}$ is secure if the attacker leverages the SM2 key to complete the session. Apparently, it's desirable for $\mathsf{tpm.KE}$ to guarantee the security of $\hat{A}$'s session key. We denote this security property by \textit{correspondence property}. The requirements for analyzing multi-protocols simultaneously and capturing the \textit{correspondence property} lead to our third motivation:
\begin{enumerate}\setcounter{enumi}{2}
\item \textit{Can we build a unified security AKE model, based on which we can give a formal analysis of $\mathsf{tpm.KE}$ which supports three AKE protocols?}
\end{enumerate}

Unfortunately, even if we prove $\mathsf{tpm.KE}$ is secure, $\mathsf{tpm.KE}$ only guarantees its security under the following conditions: first, all the entities in the network must use TPM chips to run the AKE protocols and if one entity uses a less secure implementation of protocols such as a software implementation, the compromise of the software implementation will affect the security of other honest entities, such as Kaliski's and Xu's UKS attacks mentioned above; second, attackers are assumed to be unable to obtain the information inside the TPM, such as long-term keys. The conditions are due to the fact that the above security model creates all protocol instances based on the TPM and prohibits attackers from obtaining information inside the TPM even they compromise the TPM. However, the above conditions are not practical for real-world networks: first, it is unrealistic to assume that all devices are protected by the TPM, and there may exist some devices whose execution environments are not secure and can be easily compromised; second, it is also unrealistic to assume that no TPM is ever compromised because attackers can launch sophisticated attacks such as invasive attacks, semi-invasive attacks or side-channel attacks to compromise TPM chips. The above unrealistic assumptions lead to our fourth motivation:
\begin{enumerate}\setcounter{enumi}{3}
\item \textit{How to revise the current version of $\mathsf{tpm.KE}$ so that it can be applied in real-world networks where devices may not be protected by TPM chips and TPM chips may be compromised?}
\end{enumerate}

\textbf{Contributions.} We summarize the contributions of this paper as follows:
\begin{enumerate}
\item We leverage the min-entropy, a notion in the information theory, to quantitatively evaluate the amount of randomness in the output distributions of $avf()$ and $avf'()$. We evaluate several series of elliptic curves used in practice, covering all elliptic curves adopted by the TPM 2.0 algorithm specification \cite{TPM2.0alg}. The evaluation results show that $avf()$ and $avf'()$ provide almost the same level of randomness as cryptographic hash functions.
\item We model the protections provided by the TPM by modeling the interfaces of $\mathsf{tpm.KE}$ as oracles, and present a unified AKE security model for $\mathsf{tpm.KE}$, which captures not only the basic security property defined by modern AKE security models but also the correspondence property.
\item We give a formal analysis of $\mathsf{tpm.KE}$ in our new model and prove that $\mathsf{tpm.KE}$ is secure under the condition that the unhashed shared secrets are not available to the attacker. This condition can be achieved by slightly modifying the Full UM functionality of TPM 2.0 or properly implementing the host's software.
\item The $\mathsf{tpm.KE}$ is proven secure under some limiting conditions, resulting in some restrictions on the usage of $\mathsf{tpm.KE}$, so we give suggestions on how to use $\mathsf{tpm.KE}$ properly to achieve a secure implementation of AKE protocols.
\item The limiting conditions required by the current version of $\mathsf{tpm.KE}$ are impractical for real-world networks, so we present a revision of $\mathsf{tpm.KE}$ to eliminate the limiting conditions and give concrete suggestions on how to revise the current version of TPM 2.0 specifications.
\item We rigorously analyze our revision of $\mathsf{tpm.KE}$, and the analysis results show that our revision achieves not only the basic security property defined by modern AKE security models but also some further security properties: resistance to key-compromise impersonation (KCI) and weak Perfect Forward Secrecy (PFS).
\end{enumerate}

\subsection{Organization}
In the rest of this paper, Section \ref{sec:pre} gives some preliminaries. Section \ref{sec:informal} introduces the two-phase key exchange primitive defined by TPM 2.0 specifications, gives a quantitative measure of several series of elliptic curves used in practice, and presents an informal analysis of $\mathsf{tpm.KE}$. Section \ref{sec:model} presents our unified security model for $\mathsf{tpm.KE}$. Section \ref{sec:formaldes} gives a formal description of $\mathsf{tpm.KE}$. Section \ref{sec:unforge} proves the unforgeability of the MQV and SM2 key exchange functionalities provided by $\mathsf{tpm.KE}$, and it can simplify our security proof. Section \ref{sec:analysis} formally analyzes the basic security of $\mathsf{tpm.KE}$ in our new model. Section \ref{sec:further-tpm.ke} discusses the KCI-resistance and weak PFS properties of $\mathsf{tpm.KE}$. Section \ref{sec:sug-tpm.ke} gives suggestions on the usage of the current version of $\mathsf{tpm.KE}$. Section \ref{sec:tpm.ke.rev} presents our revision of $\mathsf{tpm.KE}$. Sections \ref{sec:pre-tpm.ke.rev} and \ref{sec:analysis-rev} rigorously analyze our revision of $\mathsf{tpm.KE}$. Section \ref{sec:tpm.ke.rev-further} discusses the KCI-resistance and weak PFS properties of our revision. Section \ref{sec:conclusion} concludes this paper.

\section{Preliminaries}\label{sec:pre}
This section first introduces the notion of min-entropy and two popular methods to calculate the min-entropy, and then introduces the CDH (Computational Diffie-Hellman) and GDH (Gap Diffie-Hellman) assumptions used in this paper.

\subsection{Min-entropy}
Min-entropy is a notion in information theory, which provides a very strict information-theoretical lower bound (i.e., worst-case) measure of randomness for a random variable. High min-entropy indicates that the distribution of the random variable is close to the uniform distribution. Low min-entropy indicates that there must be a small set of outcomes that has an unusually high probability, and the small set can help an attacker to perform group representation attacks. Take the following two extreme cases for example: if the min-entropy of a random variable is equal to the length of the outcome, the distribution is a uniform distribution, and if the min-entropy of a random variable is zero, the outcomes of the random variable are a fixed value. From the two extreme cases, we can see that the higher the min-entropy is, the harder for the attacker to mount group representation attacks. There are two popular methods to measure the min-entropy of a random variable:
\begin{enumerate}
\item NIST SP 800-90. This method is described in NIST specification 800-90 for binary sources. The definition of min-entropy for one binary bit is:  $H=-log_2^{p_{max}}$ where $p_{max}=max\{p_0,p_1\}$ and $p_0$, $p_1$ are probabilities that the binary bit outputs zero and one respectively. The min-entropy of an $n$-bit binary string is defined by:
    \begin{equation}\label{equ:entory}
    H_{total}=\sum_{i=1}^{n}H_i
    \end{equation}
\item Context-Tree Weighting compression. Context-Tree Weighting (CTW) \cite{CTW} is an optimal compression algorithm for stationary sources and is usually used to estimate the min-entropy.
\end{enumerate}

\subsection{CDH and GDH Assumptions}
The security of modern cryptographic constructions, including the AKE protocols, relies on some widely-believed assumptions, which are mathematical problems that are hard to solve. The rigorous proof of security typically shows how to reduce the security of cryptographic constructions to assumptions. Here we list the assumptions that used in our proof: the CDH and GDH assumptions. The CDH assumption is used in the proof of the unforgeability of MQV and SM2 key exchange functionalities (Section \ref{sec:unforge}), and the GDH assumption is used in the proof of $\mathsf{tpm.KE}$ (Section \ref{sec:analysis}) and the revision of $\mathsf{tpm.KE}$ (Sections \ref{sec:analysis-rev} and \ref{sec:tpm.ke.rev-further}).

\newtheorem{definition}{Definition}
\begin{definition}[CDH Assumption]
Let $G$ be a cyclic group of order $p$ with generator $g$. The CDH assumption in $G$ states that, given two randomly chosen points $X=g^x$ and $Y=g^y$, it is computationally infeasible to compute $Z=g^{xy}$.
\end{definition}

\begin{definition}[GDH Assumption]
Let $G$ be a cyclic group generated by an element $g$ whose order is $p$. We say that a decision algorithm $\mathcal{O}$ is a Decisional Diffie-Hellman (DDH) Oracle for $G$ and its generator $g$ if on input a triple $(X,Y,Z)$, for $X,Y \in G$, oracle $\mathcal{O}$ outputs 1 if and only if $Z$ = CDH($X,Y$). We say that $G$ satisfies the GDH assumption if no feasible algorithm can solve the CDH problem, even if the algorithm is provided with a DDH-oracle for $G$.
\end{definition}

\section{The TPM Key Exchange Primitive}\label{sec:informal}
In the context of $\mathsf{tpm.KE}$, each party consists of a host and a TPM chip.
The host and the TPM together can implement the three implicitly AKE protocols supported in the TPM 2.0 specifications. In general, the host runs instances of AKE protocols by exchanging messages between parties and invoking the interfaces of $\mathsf{tpm.KE}$ which are implemented as TPM commands, and at last derives session keys for instances from the shared secrets returned by the TPM. In this section, we first show how the AKE protocols can be implemented by the host and $\mathsf{tpm.KE}$, and describe the related TPM commands of $\mathsf{tpm.KE}$, and then give an informal analysis of $\mathsf{tpm.KE}$. In the informal analysis, we present solutions to prevent impersonation attacks on the Full UM protocol, and a quantitative measure of the randomness of the output distributions of $avf()$ and $avf'()$ on all the elliptic curves adopted by the TPM 2.0 specifications.

\subsection{Introduction of tpm.KE}
$\mathsf{tpm.KE}$ consists of two phases. In the first phase, the TPM generates an ephemeral key which will be transferred to the other party by the TPM's host. In the second phase, the TPM generates the unhashed secret value according to the specification of the selected protocol, and the host derives the session key from the unhashed secret value. Before the two phases, the \textbf{Key Generation} procedure should be invoked to generate a long-term key.

\begin{itemize}
\item \textbf{Key Generation}. The commands $\mathsf{TPM2\_Create()}$ and $\mathsf{TPM2\_CreatePrimary()}$ are used to generate long-term keys. They take as input public parameters including an attribute identifying the key exchange scheme for the long-term key. The scheme should be one of the following three: $\mathsf{TPM\_ALG\_ECDH}$, $\mathsf{TPM\_ALG\_ECMQV}$, and $\mathsf{TPM\_ALG\_SM2}$. In this procedure, the TPM performs the following steps: if the command is $\mathsf{TPM2\_Create()}$, it picks a random $a\in_{R}Z_q$ and computes $A=g^a$, and if the command is $\mathsf{TPM2\_CreatePrimary()}$, it derives $a$ from a primary seed using a key derivation function and computes $A=g^a$; finally, it returns $A$ and a key handle identifying $a$.
\item \textbf{First Phase.} The command $\mathsf{TPM2\_EC\_Ephemeral()}$ is used to generate an ephemeral key, and it performs the following steps:
    \begin{enumerate}
    \item Generate $x=\mathsf{KDFa}(Random, Count)$, where $\mathsf{KDFa}()$ is a key derivation function \cite{KDF}, $Random$ is a secure random value stored inside the TPM, and $Count$ is a counter.
    \item Set $ctr=Count$, $A[ctr]=1$, and $Count=Count+1$, where $A[]$ is an array of bits used to indicate whether the ephemeral key has been used.
    \item Set $x=x$ mod $q$, and generate $X=g^x$.
    \item Return $X$ and $ctr$.
    \end{enumerate}
    Note that the TPM does not need to store the ephemeral private key $x$ as it can be recovered using $\mathsf{KDFa()}$ and $ctr$.
\item \textbf{Second Phase.} The command related to this phase is $\mathsf{TPM2\_ZGen\_2Phase()}$, and it is the main command of $\mathsf{tpm.KE}$. This command takes the following items as input:\\
    \begin{tabular}{lp{0.8\columnwidth}}
    $scheme$ & A protocol scheme selector. \\
    $keyA$ & The key handle identifying the long-term private key $a$.\\
    $ctr$  & The counter used to identify the ephemeral key generated in the first phase.\\
    $B$    & The public key of $\hat{B}$, with whom $\hat{A}$ wants to establish a session.\\
    $Y$    & The ephemeral public key received from $\hat{B}$.\\
    \end{tabular}
    \begin{enumerate}
    \item The TPM first does the following checks:
    \begin{enumerate}
    \item Whether $scheme$ equals the scheme designated for key $A$ in the key generation procedure.
    \item Whether $B$ and $Y$ are on the curve associated with $A$.
    \item Whether $A[ctr]=1$.
    \end{enumerate}
    \item If the above checks succeed, the TPM recovers $x=\mathsf{KDFa}(Random, ctr)$ and performs the following steps:
    \begin{enumerate}
    \item Compute unhashed values according to the value of $scheme$:
        \begin{description}
        \item Case $\mathsf{TPM\_ALG\_ECDH}$: \\
            set $Z_1=B^a$, $Z_2=Y^x$;
        \item Case $\mathsf{TPM\_ALG\_ECMQV}$: \\
            set $Z_1=(YB^e)^{h(x+da)}$, $Z_2 = NULL$, where $d=avf(X)$ and $e=avf(Y)$;
        \item Case $\mathsf{TPM\_ALG\_SM2}$: \\
            set $Z_1=(BY^e)^{h(a+dx)}$, $Z_2 = NULL$, where $d=avf'(X)$ and $e=avf'(Y)$.
        \end{description}
    \item Clear $A[ctr]$: set $A[ctr]=0$ to ensure that the ephemeral private key $x$  can only be used once.
    \item Return $Z_1$ and $Z_2$.
    \end{enumerate}
     \item Finally, the host computes the session key leveraging the unhanshed values $Z_1$ and $Z_2$ returned by the TPM.

    \end{enumerate}
\end{itemize}

\subsection{Informal Analysis}
In Sections \ref{subsec:weak-AKE} and \ref{subsec:motivation}, we show two weaknesses of $\mathsf{tpm.KE}$ which prevent it from achieving the basic security property defined by modern AKE security models. One weakness is that $\mathsf{tpm.KE}$ returns $Z_1$ of the Full UM protocol to the host whose memory can be compromised, which makes $Z_1$ be available to the attacker. The other one is the weakness caused by $avf()$ and $avf'()$, which results in group representation attacks against the MQV and SM2 key exchange protocols.

\subsubsection{Analysis of the first weakness}
The first weakness prevents the Full UM protocol from being proven secure in the security model, and we give two solutions to overcome the weakness:

\begin{enumerate}
\item Perform the entire session key computation of Full UM in the secure environment of the TPM, i.e., modify the $\mathsf{TPM2\_ZGen\_2Phase()}$ command not to return $Z_1$ and $Z_2$ but the session key, i.e., $K=H_1(Z_1,Z_2,\hat{A},\hat{B},X,Y)$.
\item Protect $Z_1$ and $Z_2$ from malicious code running on the host as much as possible such as keeping them only available in the kernel mode, and delete $Z_1$ and $Z_2$ as soon as the session key is derived.
\end{enumerate}
The first solution requires modifying TPM 2.0 specifications, and the second one requires that the software deriving the session key should be implemented properly and included in the Trusted Computing Base (TCB).

\subsubsection{Analysis of the second weakness}
As it seems that the second weakness only happens in theory, we perform a quantitative measure of the min-entropy contained in the output distributions of $avf()$ and $avf'()$ to check whether this weakness can happen in the real world. We measure several series of widely deployed elliptic curves: the NIST series \cite{NIST186-2}, the BN series \cite{BN}, the SECG series \cite{SECG}, and an SM2 elliptic curve \cite{SM2curve}. Our measure totals 17 elliptic curves and covers all elliptic curves adopted by TPM 2.0 \cite{TPM2.0alg}. We generate 16384 points for each elliptic curve, apply $avf'()$ to the points of the SM2 P256 curve, and apply $avf()$ to the points of the rest curves. We also apply the cryptographic hash function SHA-2 to the generated points of all curves. Then we measure the min-entropy of the output distributions of $avf()$ ($avf'()$) and SHA-2 using the methods of NIST SP 800-90 (Formula \ref{equ:entory}) and CTW compression. The measurement results are summarized in Table \ref{tab:statistic}. Figure \ref{fig:statis} shows the development of the min-entropy value calculated using NIST's method over the number of measurements. To our surprise, the min-entropy of the output distributions of $avf()$ and $avf'()$ is very close to the min-entropy of the output distribution of SHA-2: the former is only about 1 bit less than the latter. What's more, the measurement results indicate that the output distributions of $avf()$ and $avf'()$ are close to the uniform distribution. Take the measurement of BN P256 for example, the min-entropy calculated by the NIST's method is 126.9, very close to the output length of $avf()$ which is $129=\lceil 256/2 \rceil$+1, and the CTW ratio is 98.1\% which is close to 100\%. Our measure indicates that the outputs of $avf()$ ($avf'()$) on different elliptic curve points are almost independent, and it is infeasible to mount group representation attacks on real-world elliptic curves. So we model $avf()$ and $avf'()$ as random oracles in our formal analysis.

\renewcommand\arraystretch{1}
\begin{table}
\tiny
\centering
\begin{tabular}{|c|c|c|c|c|c|c|c|c|c|c|c|c|c|c|c|c|c|c|} \hline
\multicolumn{2}{|c|}{}                  &\multicolumn{5}{c|}{NIST Series}   &\multicolumn{6}{c|}{BN Series}             &\multicolumn{5}{c|}{SECG Series}           &SM2\\ \cline{3-19}
\multicolumn{2}{|c|}{}                  &P192 &P224 &P256 &P384  &P521      &P192 &P224 &P256 &P384  &P521 &P638        &P192 &P224 &P256 &P384  &P521              &P256 \\ \hline
NIST    &                    $avf()$  &95.2&111.0&126.9&190.2&258.7         &95.1&111.0&126.9&190.2&253.6&315.0         &95.2&111.0&126.6&190.4&258.6          &125.8  \\ \cline{2-19}
800-89  &                    SHA-2    &95.9&112.0&127.9&191.3&259.8        &96.2&112.0&128.0&191.2&254.8&316.1         &96.0&11.0&127.9&191.3&259.6          &126.9  \\ \hline
CTW     &                    $avf()$  &97\%&98\%&98\%&99\%&100\%            &97\%&98\%&98\%&99\%&99\%&100\%             &97\%&98\%&98\%&99\%&100\%   &100\%  \\ \cline{2-19}
Ratio   &                    SHA-2    &98\%&98\%&99\%&99\%&100\%            &98\%&98\%&99\%&99\%&100\%&100\%            &98\%&98\%&99\%&99\%&100\%   &100\%  \\ \hline
\end{tabular}
\caption{Min-entropy results}\label{tab:statistic}
\end{table}

\begin{figure*}[htbp]
\begin{center}
\includegraphics[width=6in]{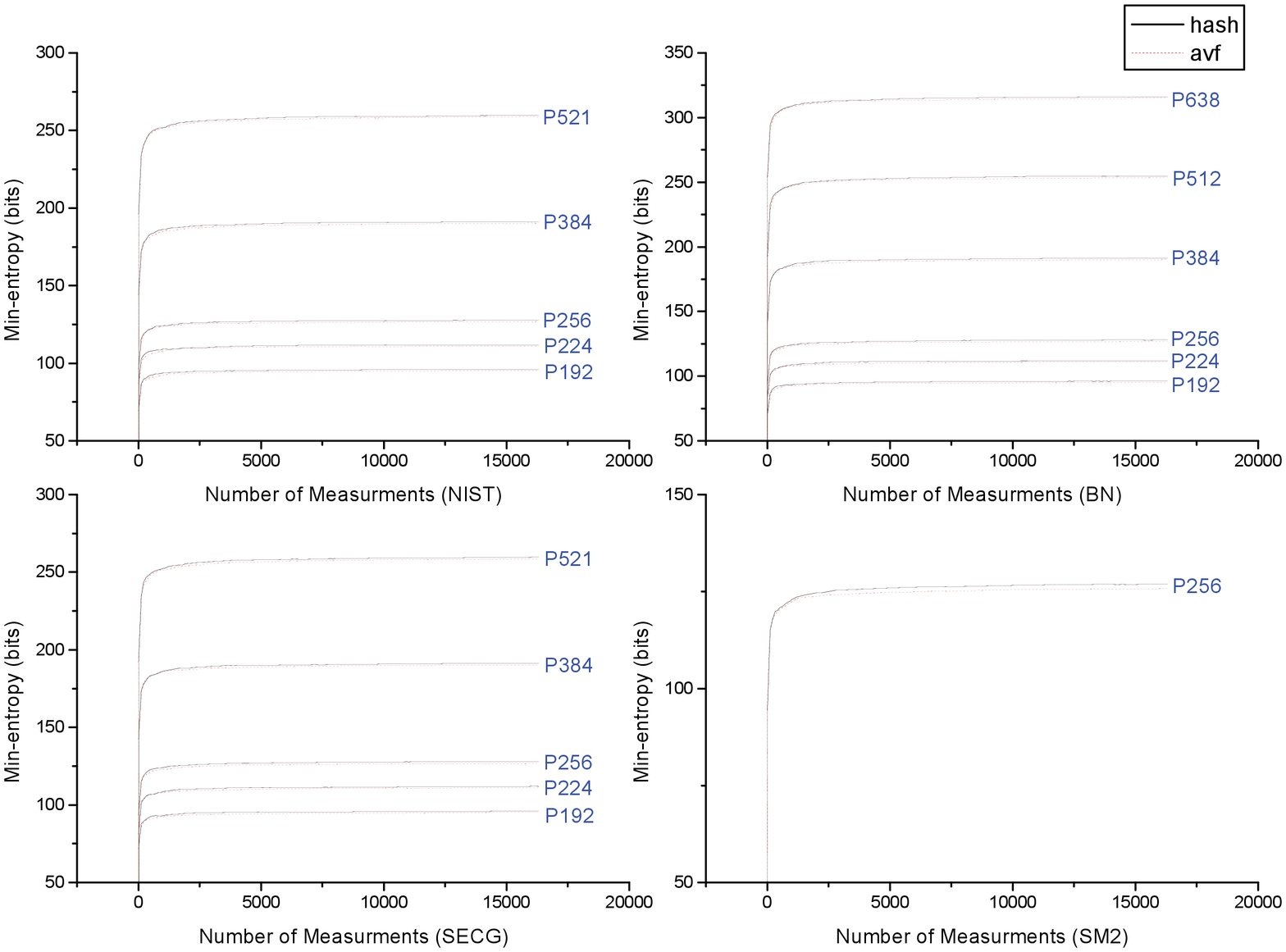}
\end{center}
\caption{\label{fig:statis}Min-entropy evaluation}
\end{figure*}

\section{The Unified Security Model}\label{sec:model}
This section illustrates our unified security model for $\mathsf{tpm.KE}$ and the attacker model which captures the capabilities of attackers.
The first difference between our unified security model and existing security models, such as BR, CK, and eCK models, is that our model can simulate multiple types of AKE protocols simultaneously. Our model acquires this feature by adding a scheme identifier to the session identifier. The motivation and detailed techniques of this feature will be discussed later. The second difference is that our model models the concrete protections provided by the TPM chip: the attack queries of our model do not allow the attacker to obtain the cryptographic keys protected inside the TPM nor the ephemeral secrets inside the TPM. For instance, when the attacker corrupts a party, unlike existing security models, which allow him to get plaintext of the long-term key, our model only allows him to obtain the black-box access of the TPM commands in respect of the long-term key.

In our security model, each party has a long-term key generated by the TPM and a certificate (issued by a CA) that binds the public part of the long-term key to the identity of the party. The long-term key can be one of the following three types: $\mathsf{TPM\_ALG\_ECDH}$, $\mathsf{TPM\_ALG\_ECMQV}$, and $\mathsf{TPM\_ALG\_SM2}$. A party can be activated to invoke the interfaces of $\mathsf{tpm.KE}$ to run an instance of the protocol supported by the long-term key, and an instance of a protocol is called a session. In each session, a party can be activated as the role of initiator or responder who sends an ephemeral public key to its peer party, the ephemeral key is generated by invoking the interface of the first phase of $\mathsf{tpm.KE}$, and a party can complete the session by invoking the interface of the second phase of $\mathsf{tpm.KE}$ and computing the session key.

In modern AKE security models, the session identifier is a quadruple $(\hat{A}, \hat{B}, Out, In)$, where $\hat{A}$ is the identity of the owner of the session, $\hat{B}$ the peer party, $Out$ the outgoing messages in the session, and $In$ the incoming messages. This kind of security models work well in proving a single protocol, but it is flawed when used to prove $\mathsf{tpm.KE}$ because $\mathsf{tpm.KE}$ supports more than one scheme (type of protocol) and instances running different schemes might have the same quadruple $(\hat{A}, \hat{B}, Out, In)$ but should be recognized as different sessions. To deal with this problem, we use a quintuple $(sc, \hat{A}, \hat{B}, Out, In)$ as the session identifier where $sc$ denotes the scheme of the session. This kind of session identifier not only distinguishes between different instances of the same scheme but also distinguishes between different schemes. For instance, two instances with different original quadruple $(\hat{A}, \hat{B}, Out, In)$, such as the two sessions $(sc, \hat{A}, \hat{B}, Out, In)$ and $(sc, \hat{A}, \hat{B}, Out', In)$ having different outgoing messages, are recognized as different sessions. And further more, two instances with same original quadruple $(\hat{A}, \hat{B}, Out, In)$ but running different schemes, such as the two sessions $(sc, \hat{A}, \hat{B}, Out, In)$ and $(sc', \hat{A}, \hat{B}, Out, In)$, are also recognized as different sessions. In particular, in the case of implicitly AKE protocols where the communication is basic Diffie-Hellman exchanges the session identifier is $(\hat{A}, \hat{B}, X, Y)$ where $X$ is the outgoing DH value and $Y$ the incoming DH value to the session. The session $(sc, \hat{B}, \hat{A}, Y, X)$ (if it exists) is said to be \textbf{matching} to the session $(sc, \hat{A}, \hat{B}, X, Y)$, and the session $(sc', \hat{B}, \hat{A}, Y, X)$ where $sc'\neq sc$ (if it exists) is said to be \textbf{message-matching} to the session $(sc, \hat{A}, \hat{B}, X, Y)$.

The $sc$ in the session identifier brings an issue we must address: how about the security of the session $(sc, \hat{A}, \hat{B}, X, Y)$ if it has a corrupted message-matching session? Previous AKE security models do not capture this attack because they do not support formal analysis of multiple kinds of protocols in a unified way. However, this attack can happen on $\mathsf{tpm.KE}$ because it supports three AKE schemes and the TPM 2.0 specifications do not force the TPM to check the key type of its peer party. We say $\mathsf{tpm.KE}$ satisfies \emph{correspondence property} if it can resist the above attack, i.e., the session $(sc, \hat{A}, \hat{B}, X, Y)$ is secure if its message-matching session is compromised.

\textbf{Attacker Model.}
The model involves multiple honest parties and an attacker $\mathcal{M}$ connected via an unauthenticated network. The attacker is modeled as a probabilistic Turing machine and has full control of the communications between parties. $\mathcal{M}$ can intercept and modify messages sent over the network. $\mathcal{M}$ also schedules all session activations and session-message delivery. In addition, in order to model potential disclosure of secret information, the attacker is allowed to access secret information via the following queries:
\begin{itemize}
\item \textbf{SessionStateReveal(s)}: $\mathcal{M}$ directly queries at session $s$ while still incomplete and learns the session state for $s$. In our analysis, the session state includes the values returned by interfaces of $\mathsf{tpm.KE}$ and intermediate information stored and computed in the host.
\item \textbf{SessionKeyReveal(s)}: $\mathcal{M}$ obtains the session key for the session $s$.
\item \textbf{Corruption($\hat{P}$)}: In other AKE security models, this query allows $\mathcal{M}$ to learn the plaintext of the long-term private key of party $\hat{P}$. In our model, $\mathcal{M}$ does not learn anything about the plaintext of the long-term private key but obtains the black-box access of the key via TPM interfaces.
\item \textbf{Test(s)}: Pick $b\xleftarrow{R}{0,1}$. If $b$ = 1, provide $\mathcal{M}$ with the session key; otherwise provide $\mathcal{M}$ with a random value $r\in_R\{0,1\}^{\lambda}$. This query can only be issued to a session that is ``clean". A completed session is ``clean" if the session as well as its matching session (if it exists) is not subject to the above three queries. A session is called ``exposed" if $\mathcal{M}$ performs any one of the above three queries to this session.
\end{itemize}

Note that our model differs from previous AKE security models in that the \textbf{Corruption} query to some party does not provide the attacker with the plaintext of the long-term private key of the party but the black-box access of the long-term key which is randomly generated and protected by the TPM. This difference captures the protections provided by the TPM, which are described in Section \ref{subsec:motivation}.

The security is defined based on an experiment played by $\mathcal{M}$, in which $\mathcal{M}$ is allowed to activate sessions and perform SessionStateReveal, SessionKeyReveal, and Corruption queries. At some time, $\mathcal{M}$ performs the Test query to a clean session of its choice and gets the value returned by Test. After that, $\mathcal{M}$ continues the experiment but is not allowed to expose the test session and its matching session (if it exists). Eventually, $\mathcal{M}$ outputs a bit $b'$ as its guess, then halts. $\mathcal{M}$ wins the game if $b'=b$. The attacker with the above capabilities is called a \textbf{KE-attacker}. The formal security is defined as follows.

\begin{definition}\label{def:security}
The advantage of any KE-attacker $\mathcal{M}$ in the above experiment with $\mathsf{tpm.KE}$ is defined as
\begin{center}
$\mathsf{Adv}^{\mathsf{tpm.KE}}(\mathcal{M})$ = $\left | Pr[b'=b]-\frac{1}{2} \right |$.
\end{center}
$\mathsf{tpm.KE}$ is called secure if the following properties hold for any KE-attacker $\mathcal{M}$.
\begin{enumerate}
\item When two uncorrupted parties complete matching sessions, they output the same session key, and
\item If no efficient attacker $\mathcal{M}$ has more than a negligible advantage in winning the above experiment.
\end{enumerate}
\label{def:tpmke}
\end{definition}

The first condition is a ``consistency" requirement for sessions completed by two uncorrupted parties. The second condition is the core property for the security of $\mathsf{tpm.KE}$: it guarantees that exposure of one session does not help the attacker to compromise the security of another session. Note that our security definition of $\mathsf{tpm.KE}$ allows the attacker to expose the message-matching session, that is to say, the test session is still secure even if the message-matching session is exposed by the attacker. Thus, our model captures the \emph{correspondence property}.

The Definition \ref{def:security} defines the basic security property that a modern AKE protocol should achieve, and it captures the common UKS attacks.
We illustrate how the UKS attacks are captured in our security model. In a UKS attack, a party $\hat{A}$ ends up believing he shares a key with party $\hat{B}$, and $\hat{B}$ ends up mistakenly believing the key is shared with the attacker $\mathcal{M}$. Suppose the identifier of $\hat{A}$'s session is $(sc, \hat{A}, \hat{B}, X, Y)$ and that of $\hat{B}$'s session is $(sc, \hat{B}, \mathcal{M}, X', Y')$. The two sessions apparently are not either identical or matching, so according to Definition \ref{def:security}, the attacker $\mathcal{M}$ can expose the session $(sc, \hat{B}, \mathcal{M}, X', Y')$ to attack $\hat{A}$'s session
$(sc, \hat{A}, \hat{B}, X, Y)$. The above illustration shows that if a protocol is proven secure under the unified security model, it should resist UKS attacks.

\section{Formal Description of $\mathsf{tpm.KE}$}\label{sec:formaldes}
The implementation of $\mathsf{tpm.KE}$ in the TPM 2.0 includes two phases. The first phase is to generate an ephemeral key which is transferred to the party which the owner wants to establish a session with, and the first phase is the same to all the three protocols supported by the TPM 2.0 specifications.
The second phase provides the key exchange functionalities which generate the unhashed secret values or session keys according to the specification of the protocol. Since parties or the attacker can only use $\mathsf{tpm.KE}$ in a black-box manner: they can only use $\mathsf{tpm.KE}$' functionalities by invoking its interfaces and do not know the secrets inside the TPM. To model the black-box manner of $\mathsf{tpm.KE}$, we simulate $\mathsf{tpm.KE}$ by formalizing its interfaces as a series of oracles which receive inputs and return outputs.

We model the interface of the first phase of $\mathsf{tpm.KE}$ by the oracle $\mathsf{ephem}_{A}()$ where $A$ is the long-term key of $\hat{A}$. We model the second phase of $\mathsf{tpm.KE}$ as three oracles: the Full UM, MQV, and SM2 key exchange functionalities of $\mathsf{tpm.KE}$ are modeled as oracle $\mathcal{O}^{\mathsf{EC}}_{A}$, oracle $\mathcal{O}^{\mathsf{MQV}}_{A}$, and oracle $\mathcal{O}^{\mathsf{SM2}}_{A}$ respectively.
$\mathcal{O}^{\mathsf{MQV}}_{A}$ and $\mathcal{O}^{\mathsf{SM2}}_{A}$ take as input the input of $\mathsf{TPM2\_ZGen\_2Phase()}$ and return the unhashed values. We let $\mathcal{O}^{\mathsf{EC}}_{A}$ directly return the session key but not the unhashed values $Z_1$ and $Z_2$, and this modification simulates our solutions to the first weakness of $\mathsf{tpm.KE}$. We now formally describe $\mathsf{tpm.KE}$ through the following three session activations.

\begin{enumerate}
\item Initiate$(sc, \hat{A}, \hat{B})$: $\hat{A}$ invokes $\mathsf{ephem}_{A}()$ of its TPM to obtain an ephemeral public key $X$ and an index $ctr_x$ identifying the ephemeral private key $x$ stored in the TPM, creates a local session which is identified as (an incomplete) session $(sc, \hat{A}, \hat{B}, X)$, where $sc$ is the key exchange scheme supported by the long-term key $A$, and outputs $X$ as its outgoing ephemeral public key.
\item Respond$(sc, \hat{B}, \hat{A}, X)$ ($sc$ is the scheme supported by $B$): After receiving $X$, $\hat{B}$ performs the following steps:
    \begin{enumerate}
    \item Invoke $\mathsf{ephem}_{B}()$ of its TPM to obtain an ephemeral public key $Y$ and an index $ctr_y$ identifying the ephemeral private key $y$ stored in the TPM; output $Y$ as its outgoing ephemeral public key.
    \item With input $(sc, keyB, ctr_y, A, X)$ where $keyB$ is the key handle of $B$, invoke the corresponding oracle according to the value of $sc$:
        \begin{description}
        \item Case $\mathsf{TPM\_ALG\_ECDH}$:
            Invoke $\mathcal{O}^{\mathsf{EC}}_{B}$ and set the session key $K$ to be the return result of $\mathcal{O}^{\mathsf{EC}}_{B}$.
        \item Case $\mathsf{TPM\_ALG\_ECMQV}$:
            Invoke $\mathcal{O}^{\mathsf{MQV}}_{B}$, obtain $Z_B$ from the return result, and compute the session key $K=H_2(Z_B,\hat{A},\hat{B})$.
        \item Case $\mathsf{TPM\_ALG\_SM2}$:
            Invoke $\mathcal{O}^{\mathsf{SM2}}_{B}$, obtain $Z_B$ from the return result, and compute the session key $K=H_2(Z_B,\hat{A},\hat{B})$.
        \end{description}
    \item Complete the session with identifier $(sc, \hat{B},\hat{A},Y,X)$.
    \end{enumerate}
\item Complete$(sc,\hat{A},\hat{B},X,Y)$: $\hat{A}$ checks that it has an open session with identifier $(sc,\hat{A},\hat{B},X)$, and then performs the following steps:
    \begin{enumerate}
    \item With input  $(sc, keyA, ctr_x, B, Y)$ where $keyA$ is the key handle of $A$, invoke the corresponding oracle according to the value of $sc$:
        \begin{description}
        \item Case $\mathsf{TPM\_ALG\_ECDH}$:
            Invoke $\mathcal{O}^{\mathsf{EC}}_{A}$ and set the session key $K$ to be the return result of $\mathcal{O}^{\mathsf{EC}}_{A}$.
        \item Case $\mathsf{TPM\_ALG\_ECMQV}$:
            Invoke $\mathcal{O}^{\mathsf{MQV}}_{A}$, obtain $Z_A$ from the return result, and compute the session key $K=H_2(Z_A,\hat{A},\hat{B})$.
        \item Case $\mathsf{TPM\_ALG\_SM2}$:
            Invoke $\mathcal{O}^{\mathsf{SM2}}_{A}$, obtain $Z_A$ from the return result, and compute the session key $K=H_2(Z_A,\hat{A},\hat{B})$.
        \end{description}
    \item Complete the session with identifier $(sc, \hat{A},\hat{B},X,Y)$.
    \end{enumerate}
\end{enumerate}

\section{Unforgeability of MQV and SM2 Key Exchange Functionalities}\label{sec:unforge}
In this section, we give the formal definitions of MQV and SM2 key exchange functionalities of $\mathsf{tpm.KE}$, denoted by $\mathcal{O}^{\mathsf{MQV}}_B$ and $\mathcal{O}^{\mathsf{SM2}}_B$ respectively where $B$ is the public key of the party.
In Theorem \ref{the:mqv} and Theorem \ref{the:sm2}, we formally prove the unforgeability of $\mathcal{O}^{\mathsf{MQV}}_B$ and $\mathcal{O}^{\mathsf{SM2}}_B$ with a constraint on the attacker respectively. Although the unforgeability of $\mathcal{O}^{\mathsf{MQV}}_B$ and $\mathcal{O}^{\mathsf{SM2}}_B$ does not directly prove the security analysis of $\mathsf{tpm.KE}$, it acts as a building block in the security of $\mathsf{tpm.KE}$ and greatly simplifies our security analysis. In some cases in our security proof in Section \ref{sec:analysis}, we reduce the security of $\mathsf{tpm.KE}$ to the unforgeability of $\mathcal{O}^{\mathsf{MQV}}_B$ or $\mathcal{O}^{\mathsf{SM2}}_B$: if the attacker is able to break the security of $\mathsf{tpm.KE}$, we can leverage it to forge the result of $\mathcal{O}^{\mathsf{MQV}}_B$ or $\mathcal{O}^{\mathsf{SM2}}_B$.

\begin{definition}[The MQV Functionality of $\mathsf{tpm.KE}$]
The functionality ($\mathcal{O}^{\mathsf{MQV}}_B$) is provided by a party possessing a private/public key pair $(b,B=g^b)$. A challenger, possessing a private/public key pair $(a,A=g^a)$, provides $\mathcal{O}^{\mathsf{MQV}}_B$ with a challenge $X=g^x$ ($x$ is chosen and kept secret by the challenger). With the pair $(A,X)$, $\mathcal{O}^{\mathsf{MQV}}_B$ first computes an ephemeral private/public key pair $(y, Y=g^y)$ and returns $Z=(XA^d)^{h(y+eb)}$ where $d=avf(X)$ and $e=avf(Y)$. The challenger can verify the return result $(Y,Z)$ with respect to the challenge $X$ by checking whether $Z=(YB^e)^{h(x+da)}$.
\end{definition}

\begin{definition}[The SM2 Key Exchange Functionality of $\mathsf{tpm.KE}$]
The functionality ($\mathcal{O}^{\mathsf{SM2}}_B$) is provided by a party possessing a private/public key pair $(b,B=g^b)$. A challenger, possessing a private/public key pair $(a,A=g^a)$, provides $\mathcal{O}^{\mathsf{SM2}}_B$ with a challenge $X=g^x$ ($x$ is chosen and kept secret by the challenger). With the pair $(A,X)$, $\mathcal{O}^{\mathsf{SM2}}_B$ first computes an ephemeral private/public key pair $(y, Y=g^y)$ and returns $Z=(AX^d)^{h(b+ey)}$, where $d=avf'(X)$ and $e=avf'(Y)$. The challenger can verify the return result $(Y,Z)$ with respect to the challenge $X$ by checking whether $Z=(BY^e)^{h(a+dx)}$.
\end{definition}

\newtheorem{theorem}{Theorem}
\begin{theorem}\label{the:mqv}
Under the CDH assumption, with $avf()$ modeled as a random oracle, given a challenge $X$, it is computationally infeasible for an attacker to forge a return result of $\mathcal{O}^{\mathsf{MQV}}_B$ on behalf of a challenger whose public key is $A$ under the constraint that $(a,x)$ is unknown to the attacker.
\end{theorem}

The constraint is reasonable because the TPM prevents attackers from obtaining the plaintext of keys in it.
We prove Theorem \ref{the:mqv} by showing that if an attacker $\mathcal{M}$ can forge a return result under our constraint, then we can construct a CDH solver $\mathcal{C}$ which uses $\mathcal{M}$ as a subroutine.

\begin{proof}
$\mathcal{C}$ takes as input a pair $(X,B)\in G^2$ and $(a,A=g^a)$, and it simulates $\mathcal{O}^{\mathsf{MQV}}_B$ as follows:
\begin{enumerate}
\item On receipt of the input $(A',X')$, choose $e,s\in Z_q$ randomly.
\item Let $Y'=g^s/B^e$, and set $avf(Y')=e$.
\item Choose $d$ randomly, and set $avf(X')=d$.
\item Return $(Y,Z'=(X'A^d)^{hs})$.
\end{enumerate}

If $\mathcal{M}$ successfully forges a return result $(Y,Z)$ on the pair $(A,X)$ in an experiment, then $\mathcal{C}$ obtains $Z=(XA^d)^{h(y+eb)}$ where $d=avf(X)$ and $e=avf(Y)$. Note that without the knowledge of the private key $y$ of $Y$, $\mathcal{C}$ is unable to compute CDH($X,B$). Following the Forking Lemma \cite{fork96} approach, $\mathcal{C}$ runs $\mathcal{M}$ on the same input and the same coin flips but with carefully modified answers to $avf()$ queries. Note that $\mathcal{M}$ must have queried $avf(Y)$ in its first run because otherwise $\mathcal{M}$ would be unable to compute $Z$. For the second run of $\mathcal{M}$, $\mathcal{C}$ responds to $avf(Y)$ with a value $e'\neq e$ selected uniformly at random. If $\mathcal{M}$ succeeds in the second run, $\mathcal{C}$ obtains $Z'=(XA^d)^{h(y+e'b)}$ and can compute CDH$(X,B)=(\frac{Z}{Z'})^{\frac{1}{h(e-e')}}B^{-da}$. $\hfill\blacksquare$
\end{proof}

\begin{theorem}\label{the:sm2}
Under the CDH assumption, with $avf'()$ modeled as a random oracle, given a challenge $X$, it is computationally infeasible for an attacker to forge a return result of $\mathcal{O}^{\mathsf{SM2}}_B$ on behalf of a challenger whose public key is $A$ under the constraint that $(a,x)$ is unknown to the attacker.
\end{theorem}

We omit the proof of theorem \ref{the:sm2} because it can be easily completed following the proof of theorem \ref{the:mqv}.

\section{Security Analysis of $\mathsf{tpm.KE}$}\label{sec:analysis}
Although our model is based on the CK model, the security analysis based on our new model differs from that based on the CK model in some aspects. First the \emph{Corrupt} query in our model does not let the attacker obtain the plaintext of the long-term key but only allows him to get the black-box access of the long-term key. The second difference is about the session state that allows the attacker to get by the \emph{SessionStateReveal} query. For most security analysis of AKE protocols, the session state is defined by the requirement of the proof. For instance, the security analysis of the HMQV protocol \cite{HMQV05} requires that the session state only includes the session key, which means that the \emph{SessionStateReveal} query and the \emph{SessionKeyReveal} query are the same. But for the security analysis of $\mathsf{tpm.KE}$, we need to define the session state according to the specifications of $\mathsf{tpm.KE}$: we have to put into session state all the secrets that are processed and stored on the host's memory which is easily to be tampered with. In the following we define the session state allowed to be revealed by the attacker.

\textbf{Session State.} In order to simulate the protections provided by the TPM, we specify that the state of a session stores the results returned by the TPM and the information stored in the host. For the Full UM scheme, the session state is the session key; for the MQV and SM2 key exchange schemes, the session state is the unhashed values returned by the TPM.

\begin{theorem}\label{the-main}
Under the CDH and GDH assumptions, with hash functions $H_1()$, $H_2()$, $avf()$, and $avf'()$ modeled as random oracles, $\mathsf{tpm.KE}$ is secure in the unified model.

Let $\mathsf{Adv}^{GDH}(\mathcal{C})$, $\mathsf{Adv}^{\mathcal{O}^{\mathsf{MQV}}}(\mathcal{C})$ and $\mathsf{Adv}^{\mathcal{O}^{\mathsf{SM2}}}(\mathcal{C})$ denote the advantage of an algorithm $\mathcal{C}$ in solving the GDH problem, forging the functionality of $\mathcal{O}^{\mathsf{MQV}}$, and forging the functionality of $\mathcal{O}^{\mathsf{SM2}}$. For any attacker $\mathcal{M}$ against $\mathsf{tpm.KE}$ whose security parameter is $\lambda$, and that involves at most $n$ parties and activates at most $k$ sessions, we construct a GDH solver $\mathcal{S}$, a $\mathcal{O}^{\mathsf{MQV}}$ forger $\mathcal{F}^{\mathsf{MQV}}$, and a $\mathcal{O}^{\mathsf{SM2}}$ forger $\mathcal{F}^{\mathsf{SM2}}$ such that
\begin{equation}\nonumber
\begin{split}
\mathsf{Adv}^{\mathsf{tpm.KE}}(\mathcal{M}) \leq \frac{1}{2} \cdot \max\ & \left\{  3n^2k\cdot\mathsf{Adv}^{GDH}(\mathcal{S}), \frac{9}{2}(nk)^2\cdot\mathsf{Adv}^{GDH}(\mathcal{S}), \right.\\
  &\left.\frac{9}{2}(nk)^2\cdot\mathsf{Adv}^{\mathcal{O}^{\mathsf{MQV}}}(\mathcal{F}^{\mathsf{MQV}}), 3n^2k\cdot\mathsf{Adv}^{\mathcal{O}^{\mathsf{SM2}}}(\mathcal{F}^{\mathsf{SM2}})  \right\} + \frac{1}{2} \cdot \mathcal{O}(\frac{k^2}{2^{\lambda}}).
\end{split}
\end{equation}
\end{theorem}

\begin{proof}

The proof of the above theorem follows from the definition of secure key exchange protocols outlined in Section \ref{sec:model} and the following two lemmas.

\newtheorem{lemma}[theorem]{Lemma}
\begin{lemma}\label{lemma:same}
If two parties $\hat{A}$ and $\hat{B}$ complete matching sessions, then their session keys are the same.
\end{lemma}

\begin{lemma}\label{lemma:core}
Under the CDH and GDH assumptions, there is no feasible attacker that succeeds in distinguishing the session key of an unexposed session with non-negligible probability.
\end{lemma}

Lemma \ref{lemma:same} follows immediately from the definition of matching sessions. That is, if $\hat{A}$ completes the session $(sc,\hat{A},\hat{B},X,Y)$ and $\hat{B}$ completes the matching session $(sc,\hat{B},\hat{A},Y,X)$, it's easy to verify that $\hat{A}$'s session key is the same as $\hat{B}$'s according to the specifications of the protocols (Figure \ref{fig:protocols}).

We now prove Lemma \ref{lemma:core}. Let $\mathcal{M}$  be an attacker against $\mathsf{tpm.KE}$ and $\mathit{ts}$ the test session. Let $\mathsf{COLL}$ denote the event that the random oracles produce collisions, $\mathsf{GET}$ the event that the attacker gets the session key of $\mathit{ts}$, $\mathsf{QUERY}$ the event that someone, either the attacker or some party, queries the random oracle with the same tuple $\sigma$ as that of $\mathit{ts}$, $\mathsf{E1}$ the event that the attacker himself queries the random oracle with the same tuple $\sigma$ as that of $\mathit{ts}$, $\mathsf{E2}$ the event that some party queries the random oracle with the same tuple $\sigma$ as that of $\mathit{ts}$. Observe that $Pr\left[b'=b \mid \overline{\mathsf{GET}} \wedge \overline{\mathsf{COLL}} \right] = \frac{1}{2}$, $Pr\left[b'=b \mid \mathsf{GET} \wedge \overline{\mathsf{COLL}} \right] = 1$, $Pr\left[ \mathsf{GET} \wedge \overline{\mathsf{QUERY}} \wedge \overline{\mathsf{COLL}} \right] = 0$, and $Pr\left[\mathsf{QUERY}\right] = Pr\left[\mathsf{E1} \right] + Pr\left[\mathsf{E2} \right]$, we may write
\begin{equation}\nonumber
\begin{aligned}
\mathsf{Adv}^{\mathsf{tpm.KE}}(\mathcal{M}) &= \left|  Pr[b'=b] - \frac{1}{2} \right| \\
&= \left| Pr\left[b'=b \wedge \mathsf{COLL} \right] +  Pr\left[b'=b \wedge \mathsf{GET} \wedge \overline{\mathsf{COLL}} \right]  \right.
  \left.  + Pr\left[b'=b \wedge \overline{\mathsf{GET}} \wedge \overline{\mathsf{COLL}} \right] - \frac{1}{2} \right| \\
&= \left| Pr\left[b'=b \bigm| \mathsf{COLL} \right] \cdot Pr\left[ \mathsf{COLL} \right] + Pr\left[b'=b \bigm| \mathsf{GET} \wedge \overline{\mathsf{COLL}} \right] \cdot Pr\left[ \mathsf{GET} \wedge \overline{\mathsf{COLL}} \right]  \right.  \\
&\quad \left. +  \frac{1}{2} \cdot Pr\left[ \overline{\mathsf{GET}} \wedge \overline{\mathsf{COLL}} \right] - \frac{1}{2}\right| \\
&= \left| Pr\left[b'=b \bigm| \mathsf{COLL} \right] \cdot Pr\left[ \mathsf{COLL} \right] + Pr\left[b'=b \bigm| \mathsf{GET} \wedge \overline{\mathsf{COLL}} \right] \cdot Pr\left[ \mathsf{GET} \wedge \overline{\mathsf{COLL}} \right]  \right.  \\
&\quad \left. +  \frac{1}{2} \cdot \left( 1 - Pr\left[ \mathsf{COLL} \right] - Pr\left[ \mathsf{GET} \wedge \overline{\mathsf{COLL}} \right] \right) - \frac{1}{2}\right| \\
& \leq \frac{1}{2} \cdot \left( Pr\left[ \mathsf{GET} \wedge \overline{\mathsf{COLL}} \right] + Pr\left[ \mathsf{COLL} \right] \right) \\
& \leq \frac{1}{2} \cdot \left( Pr\left[ \mathsf{GET} \wedge \mathsf{QUERY} \wedge \overline{\mathsf{COLL}} \right] + Pr\left[ \mathsf{GET} \wedge \overline{\mathsf{QUERY}} \wedge \overline{\mathsf{COLL}} \right] + Pr\left[ \mathsf{COLL} \right] \right) \\
& \leq \frac{1}{2} \cdot \left( Pr\left[ \mathsf{GET} \wedge \mathsf{E1} \wedge \overline{\mathsf{COLL}} \right] + Pr\left[ \mathsf{GET} \wedge \mathsf{E2} \wedge \overline{\mathsf{COLL}} \right] + Pr\left[ \mathsf{COLL} \right] \right)
\end{aligned}
\end{equation}

If random oracles produce no collisions, the event $\mathsf{GET} \wedge \mathsf{E1} \wedge \overline{\mathsf{COLL}}$ means that $\mathcal{M}$ himself queries the random oracle with the same tuple $\sigma$ as that of $\mathit{ts}$ and gets the session key of $\mathit{ts}$. We denote this event by the \emph{forging attack}. And if random oracles produce no collisions, the event $ \mathsf{GET} \wedge \mathsf{E2} \wedge \overline{\mathsf{COLL}}$ means that $\mathcal{M}$ gets the session key of $\mathit{ts}$ but he does not query the random oracle with the tuple $\sigma$ of $\mathit{ts}$, and that some party establishes a session $\mathit{s}'$ with the same tuple $\sigma$ of $\mathit{ts}$. Since $\mathcal{M}$ himself does not query the tuple $\sigma$ of $\mathit{ts}$ and is not allowed to reveal $\mathit{ts}$ in the experiment, he must get the session key of $\mathit{ts}$ by revealing $\mathit{s}'$. So the event $ \mathsf{GET} \wedge \mathsf{E2} \wedge \overline{\mathsf{COLL}}$ implies that $\mathcal{M}$ forces the establishment of another session $\mathit{s}'$ that has the same session key as the test session and he gets the session key of the test session by revealing $\mathit{s}'$, and we denote this event by the \emph{key-replication attack}. We summarize the two attacks as follows:
\begin{enumerate}
\item Forging attack. At some point $\mathcal{M}$ queries $H_1()$ or $H_2()$ on the same tuple $\sigma$ as the test session.
\item Key-replication attack. $\mathcal{M}$ succeeds in forcing the establishment of another session that has the same session key as the test session.
\end{enumerate}

The rest of this section will prove the infeasibility of forging attack and key-replication attack by showing that if either of the above attacks succeed with non-negligible probability then there exists an attacker against the GDH problem or a forger against the MQV functionality of $\mathsf{tpm.KE}$, or a forger against the SM2 key exchange functionality of $\mathsf{tpm.KE}$, and the latter two forgers are in contradiction to the CDH assumption (Theorem \ref{the:mqv} and Theorem \ref{the:sm2}). The proof of Lemma \ref{lemma:core} will be completed after the proof of the infeasibility of forging attack and key-replication attack.
$\hfill\blacksquare$
\end{proof}

\subsection{Infeasibility of Forging Attacks}
Consider a successful forging attack performed by $\mathcal{M}$. Let $(sc,\hat{A},\hat{B},X_0,Y_0)$ be the test session for whose tuple $\mathcal{M}$ outputs a correct guess. By the convention on session identifiers, we know that the test session is held by $\hat{A}$, its peer is $\hat{B}$, $X_0$ was output by $\hat{A}$, and $Y_0$ was the incoming message to $\hat{A}$. $sc$ can fall under one of the following three cases:
\begin{enumerate}
\item $sc=\mathsf{TPM\_ALG\_ECDH}$.
\item $sc=\mathsf{TPM\_ALG\_ECMQV}$.
\item $sc=\mathsf{TPM\_ALG\_SM2}$.
\end{enumerate}

As we assume that $\mathcal{M}$ succeeds with non-negligible probability in the forging attack, there is at least one of the above three cases that occurs with non-negligible probability. We assume that $\mathcal{M}$ operates in an environment that involves at most $n$ parties and each party participates in at most $k$ sessions. We analyze the three cases separately in the following.

\subsubsection{Analysis of Case 1}
For this case, we build a GDH solver $\mathcal{S}_1$ with the following property: if $\mathcal{M}$ succeeds with non-negligible probability in this case, then $\mathcal{S}_1$ succeeds with non-negligible probability in solving the GDH problem. $\mathcal{S}_1$ takes as input a pair $(A,B)$, creates an experiment which includes $n$ honest parties and the attacker $\mathcal{M}$, and is given access to a DDH oracle $DDH$. $\mathcal{S}_1$ randomly selects two parties $\hat{A}$ and $\hat{B}$ from the honest parties and sets their public keys to be $A$ and $B$ respectively, and all the other parties compute their keys normally. Furthermore, $\mathcal{S}_1$ randomly selects an integer $i\in[1,...,k]$. The simulation for $\mathcal{M}$'s environment proceeds as follows:
\begin{enumerate}\setcounter{enumi}{-1}
\item $\mathcal{S}_1$ models $\mathsf{ephem}_{P}()$ for all parties except party $\hat{B}$ following the description in Section \ref{sec:formaldes}. $\mathcal{M}$ sets the type of all long-term keys. If the type of $A$ is not $\mathsf{TPM\_ALG\_ECDH}$, $\mathcal{S}_1$ aborts. $\mathcal{S}_1$ creates oracles modeling the two-phase key exchange functionalities for each party normally except parties $\hat{A}$ and $\hat{B}$ because it possesses the long-term private keys of all parties except $\hat{A}$ and $\hat{B}$. $H_1()$, $H_2()$, $avf()$, and $avf'()$ are modeled as random oracles described below.
\item Initiate$(sc,\hat{P_1}, \hat{P_2})$: $\hat{P_1}$ executes the Initiate() activation of the protocol. However, if the session being created is the $i$-th session at $\hat{A}$, $\mathcal{S}_1$ checks whether $sc=\mathsf{TPM\_ALG\_ECDH}$ and $\hat{P_2} = \hat{B}$. If not, $\mathcal{S}_1$ aborts.
\item Respond$(sc,\hat{P_1}, \hat{P_2}, Y)$: With the exception of $\hat{A}$ and $\hat{B}$ (whose behaviors we explain below), $\hat{P_1}$ executes the Respond() activation of the protocol. However, if the session being created is the $i$-th session at $\hat{A}$, $\mathcal{S}_1$ checks whether  $sc=\mathsf{TPM\_ALG\_ECDH}$ and $\hat{P_2} = \hat{B}$. If not, $\mathcal{S}_1$ aborts.
\item Complete$(sc,\hat{P_1}, \hat{P_2}, X, Y)$:  With the exception of $\hat{A}$ and $\hat{B}$ (whose behaviors we explain below), $\hat{P_1}$ executes the Complete() activation of the protocol. However, if the session is the $i$-th session at $\hat{A}$, $\mathcal{S}_1$ completes the session without computing a session key.
\item With the input $(sc, keyA, ctr_x, P, Y)$, $\mathcal{S}_1$ creates the oracle $\mathcal{O}^{\mathsf{EC}}_{A}$ as follows: \label{step:modelECA}
    \begin{enumerate}
    \item If $\hat{P}=\hat{B}$, $\mathcal{O}^{\mathsf{EC}}_{A}$ returns a session key to be $H_{spec}(\hat{A},\hat{B},X,Y)$. $H_{spec}()$ is simulated as a random oracle.
    \item If $\hat{P}\neq \hat{B}$, returns a session key to be $H_1(Z_1,Z_2,\hat{A},\hat{P},X,Y)$ where $Z_1=A^p$ ($p$ is the long-term private key of $\hat{P}$) and $Z_2=Y^x$ ($x$ is the ephemeral private key indexed by $ctr_x$).
    \end{enumerate}
\item Now $\mathcal{S}_1$ can simulate all the session activations at $\hat{A}$ for $\mathcal{M}$ with the help of $\mathcal{O}^{\mathsf{EC}}_{A}$.
\item $\mathcal{S}_1$ creates a Table $T$ and models $\mathsf{ephem}_{B}()$ according to the type of $B$:
    \begin{enumerate}
    \item Case $\mathsf{TPM\_ALG\_ECDH}$: Model following the description in Section \ref{sec:formaldes}.
    \item Case $\mathsf{TPM\_ALG\_ECMQV}$:
        \begin{enumerate}
        \item Randomly choose $e,s\in Z_q$.
        \item Set $Y=g^s/B^e$ and $e=avf(Y)$.
        \item Randomly choose an index $ctr$, and add a record $(ctr, e,s,Y,-)$ to $T$.
        \item Return $ctr$ and $Y$.
        \end{enumerate}
    \item Case $\mathsf{TPM\_ALG\_SM2}$:
        \begin{enumerate}
        \item Randomly choose $e,s\in Z_q$.
        \item Set $Y=(g^s/B)^{e^{-1}}$ and $e=avf'(Y)$.
        \item Randomly choose an index $ctr$, and add a record $(ctr, e,s,Y,-)$ to $T$.
        \item Return $ctr$ and $Y$.
        \end{enumerate}
    \end{enumerate}

\item With the input $(sc, keyB, ctr_y, P, X)$, $\mathcal{S}_1$ creates the oracle $\mathcal{O}^{\mathsf{EC}}_{B}$ to model the two-phase key exchange functionality for party $\hat{B}$ according to the type of $B$: \label{step:oracleB}
    \begin{enumerate}
    \item Case $\mathsf{TPM\_ALG\_ECDH}$: $\mathcal{O}^{\mathsf{EC}}_{B}$ is modeled similarly to $\mathcal{O}^{\mathsf{EC}}_{A}$ which is described in step \ref{step:modelECA}.
    \item Case $\mathsf{TPM\_ALG\_ECMQV}$:
        \begin{enumerate}
        \item Check whether (1) $sc=\mathsf{TPM\_ALG\_ECMQV}$, (2) $P$ and $X$ are on the curve associated with $B$, and (3) the last element of the record in $T$ indexed by $ctr_y$ is `$-$'. If the above checks succeed, continue, else return an error.
        \item Suppose the record in $T$ indexed by $ctr_y$ is $(ctr_y, e,s,Y,-)$, set $Z_1=(XP^d)^{hs}$ where $d=avf(X)$, and set the last element of the record to be $\times$.
        \item Return $(Z_1, Z_2=NULL)$.
        \end{enumerate}
    \item Case $\mathsf{TPM\_ALG\_SM2}$:
        \begin{enumerate}
        \item Check whether (1) $sc=\mathsf{TPM\_ALG\_SM2}$, (2) $P$ and $X$ are on the curve associated with $B$, and (3) the last element of the record in $T$ indexed by $ctr_y$ is `$-$'. If the above checks pass, continue, else return an error.
        \item Suppose the record in $T$ indexed by $ctr_y$ is $(ctr_y, e,s,Y,-)$, set $Z_1=(PX^d)^{hs}$ where $d=avf'(X)$, and set the last element of the record to be $\times$.
        \item Return $(Z_1, Z_2=NULL)$.
        \end{enumerate}
    \end{enumerate}
\item $\mathcal{S}_1$ simulates all the session activations at $\hat{B}$ for $\mathcal{M}$ with the help of $\mathsf{ephem}_{B}()$ and the oracle created in step \ref{step:oracleB}.
\item SessionStateReveal($s$): $\mathcal{S}_1$ returns to $\mathcal{M}$ the session state of session $s$. However, if $s$ is the $i$-th session at $\hat{A}$, $\mathcal{S}_1$ aborts.
\item SessionKeyReveal($s$): $\mathcal{S}_1$ returns to $\mathcal{M}$ the session key of $s$. If $s$ is the $i$-th session at $\hat{A}$, $\mathcal{S}_1$ aborts.
\item Corruption($\hat{P}$): $\mathcal{S}_1$ gives $\mathcal{M}$ the handle of the long-term key $P$. If $\mathcal{M}$ tries to corrupt $\hat{A}$ or $\hat{B}$, $\mathcal{S}_1$ aborts.
\item $H_1(\sigma)$ function for some $\sigma=(Z_1,Z_2,\hat{P_1},\hat{P_2},X,Y)$ proceeds as follows:
    \begin{enumerate}
    \item If $\hat{P_1}=\hat{A}$, $\hat{P_2}=\hat{B}$, and $DDH(A,B,Z_1)=1$, then $\mathcal{S}_1$ aborts $\mathcal{M}$ and succeeds by outputting CDH$(A,B)=Z_1$.
    \item If the value of the function on input $\sigma$ has been defined previously, return it.
    \item If the value of $H_{spec}()$ on input $(\hat{P_1},\hat{P_2},X,Y)$ has been defined previously, return it.
    \item Pick a random key $k\in_R \{0,1\}^{\lambda}$, and define $H_1(\sigma)=k$.
    \end{enumerate}
\item $H_{spec}()$, $H_2()$, $avf()$, and $avf'()$ are simulated as random oracles in the usual way.
\end{enumerate}

\begin{proof}
The probability that $\mathcal{M}$ sets the type of $A$ to be $\mathsf{TPM\_ALG\_ECDH}$ and selects the $i$-th session of $\hat{A}$ and the peer of the test session is party $\hat{B}$ is at least $\frac{1}{3n^2k}$. Suppose that this indeed the case: the type of $A$ is $\mathsf{TPM\_ALG\_ECDH}$, so $\mathcal{S}_1$ does not abort in Step 0; $\mathcal{M}$ is not allowed to corrupt $\hat{A}$ and $\hat{B}$, make SessionStateReveal and SessionKeyReveal queries to the $i$-th session of $\hat{A}$, so $\mathcal{S}_1$ does not abort in Step 1, 2, 9, 10, 11. Therefore, $\mathcal{S}_1$ simulates $\mathcal{M}$'s environment perfectly. Thus, if $\mathcal{M}$ wins with non-negligible probability in this case, the success probability of $\mathcal{S}_1$ is bounded by
\begin{center}
$Pr(\mathcal{S}_1)\geq \frac{1}{3n^2k}Pr(\mathcal{M})$.
\end{center}
$\hfill\blacksquare$
\end{proof}

\subsubsection{Analysis of Case 2}\label{subsec:case2}
Recall that the test session is denoted by $(sc,\hat{A},\hat{B},X_0,Y_0)$. We divide Case 2 of the forging attack into the following four subcases according to the generation of $Y_0$:
\begin{enumerate}[C1.]
\item $Y_0$ was generated by $\hat{B}$ in a session matching the test session, i.e., in session $(sc,\hat{B},\hat{A},Y_0,X_0)$.
\item $Y_0$ was generated by $\hat{B}$ in a session message-matching the test session, i.e., in session $(sc',\hat{B},\hat{A},Y_0,X_0)$ with $sc'\neq sc$.
\item $Y_0$ was generated by $\hat{B}$ in a session $(sc',\hat{B},\hat{A}^*,Y_0,X^*)$ with $(\hat{A}^*,X^*)\neq (\hat{A},X_0)$.
\item $Y_0$ did not appear in any completed sessions activated at $\hat{B}$, i.e., $Y_0$ was never output by $\hat{B}$ as its outgoing ephemeral public key in any sessions, or $\hat{B}$ did output $Y_0$ as its outgoing ephemeral public key for some session $s$ but it never completed $s$ by computing the session key.
\end{enumerate}

If $\mathcal{M}$ succeeds in Case 2 in its forging attack with non-negligible probability then there is at least one of the above four subcases happens with non-negligible probability in the successful runs of $\mathcal{M}$. We proceed to analyze these subcases separately.

\textbf{Analysis of Subcase C1.}
For this subcase, we build a GDH solver $\mathcal{S}_2$ with the following property: if $\mathcal{M}$ succeeds with non-negligible probability in this subcase, then $\mathcal{S}_2$ succeeds with non-negligible probability in solving the GDH problem. $\mathcal{S}_2$ takes as input a pair $(X_0,Y_0)$, creates an experiment which includes $n$ honest parties and the attacker $\mathcal{M}$, and is given access to a DDH oracle $DDH$. All parties compute their keys normally. $\mathcal{S}_2$ randomly selects two party $\hat{A}$ and $\hat{B}$ and randomly selects two integers $i,j\in[1,...,k]$. The simulation for $\mathcal{M}$'s environment proceeds as follows:
\begin{enumerate}\setcounter{enumi}{-1}
\item $\mathcal{S}_2$ models $\mathsf{ephem}_{P}()$ for all parties following the description in Section \ref{sec:formaldes}. $\mathcal{M}$ sets the type of all long-term keys. If the type of $A$ and $B$ is not $\mathsf{TPM\_ALG\_ECMQV}$, $\mathcal{S}_2$ aborts. $\mathcal{S}_2$ can create oracles modeling the two-phase key exchange functionalities for each party normally as it possesses the long-term private keys of all parties.
\item Initiate$(sc,\hat{P_1}, \hat{P_2})$: $\hat{P_1}$ executes the Initiate() activation of the protocol. However, if the session being created is the $i$-th session at $\hat{A}$ (or the $j$-th session at $\hat{B}$), $\mathcal{S}_2$ checks whether $sc=\mathsf{TPM\_ALG\_ECMQV}$ and $\hat{P_2} = \hat{B}$ (or $\hat{P_2} = \hat{A}$). If so, $\mathcal{S}_2$ sets the ephemeral public key to be $X_0$ (or $Y_0$), else $\mathcal{S}_2$ aborts.
\item Respond$(sc,\hat{P_1}, \hat{P_2}, Y)$: $\hat{P_1}$ executes the Respond() activation of the protocol. However, if the session being created is the $i$-th session at $\hat{A}$ (or the $j$-th session at $\hat{B}$), $\mathcal{S}_2$ checks whether $sc=\mathsf{TPM\_ALG\_ECMQV}$, $\hat{P_2} = \hat{B}$ (or $\hat{P_2} = \hat{A}$), and $Y=Y_0$ (or $Y=X_0$). If so, $\mathcal{S}_2$ sets the ephemeral public key to be $X_0$ (or $Y_0$), else $\mathcal{S}_2$ aborts.
\item Complete$(sc,\hat{P_1}, \hat{P_2}, X, Y)$: $\hat{P_1}$ executes the Complete() activation of the protocol. However, if the session is the $i$-th session at $\hat{A}$ (or the $j$-th session at $\hat{B}$), $\mathcal{S}_2$ completes the session without computing a session key.
\item SessionStateReveal($s$): $\mathcal{S}_2$ returns to $\mathcal{M}$ the session state of session $s$. However, if $s$ is the $i$-th session at $\hat{A}$ (or the $j$-th session at $\hat{B}$), $\mathcal{S}_2$ aborts.
\item SessionKeyReveal($s$): $\mathcal{S}_2$ returns to $\mathcal{M}$ the session key of $s$. If $s$ is the $i$-th session at $\hat{A}$ (or the $j$-th session at $\hat{B}$), $\mathcal{S}_2$ aborts.
\item Corruption($\hat{P}$): $\mathcal{S}_2$ gives $\mathcal{M}$ the handle of the long-term key $P$. If $\mathcal{M}$ tries to corrupt $\hat{A}$ or $\hat{B}$, $\mathcal{S}_2$ aborts.
\item $H_2(\sigma)$ function for some $\sigma=(Z,\hat{P_1},\hat{P_2})$ proceeds as follows:
    \begin{enumerate}
    \item If $\hat{P_1}=\hat{A}$, $\hat{P_2}=\hat{B}$, and $DDH(X_0A^d,Y_0B^e,Z^{1/h})=1$ where $d=avf(X_0)$ and $e=avf(Y_0)$, then $\mathcal{S}_2$ aborts $\mathcal{M}$ and succeeds by outputting CDH($X_0,Y_0)=\frac{Z^{1/h}}{X_0^{eb}Y_0^{da}g^{deab}}$.
    \item If the value of the function on input $\sigma$ has been defined previously, return it.
    \item Pick a random key $k\in_R\{0,1\}^{\lambda}$, and define $H_2(\sigma)=k$.
    \end{enumerate}
\item $H_1()$, $avf()$, and $avf'()$ are simulated as random oracles in the usual way.
\end{enumerate}

\begin{proof}
The probability that $\mathcal{M}$ sets the type of $A$ and $B$ to be $\mathsf{TPM\_ALG\_ECMQV}$ and selects the $i$-th session of $\hat{A}$ and the $j$-th session of $\hat{B}$ as the test session and its matching session is at least $\frac{1}{3}\times\frac{1}{3}\times\frac{2}{(nk)^2}=\frac{2}{9(nk)^2}$. Suppose that this is indeed the case: the type of $A$ and $B$ is $\mathsf{TPM\_ALG\_ECMQV}$, so $\mathcal{S}_2$ does not abort in Step 0; $\mathcal{M}$ is not allowed to corrupt $\hat{A}$ and $\hat{B}$, make SessionStateReveal and SessionKeyReveal queries to the $i$-th session of $\hat{A}$ or the $j$-th session of $\hat{B}$, so $\mathcal{S}_2$ does not abort in Step 1, 2, 4, 5, 6. Therefore, $\mathcal{S}_2$ simulates $\mathcal{M}$'s environment perfectly. Thus, if $\mathcal{M}$ wins with non-negligible probability in this case, the success probability of $\mathcal{S}_2$ is bounded by
\begin{center}
$Pr(\mathcal{S}_2)\geq \frac{2}{9(nk)^2}Pr(\mathcal{M})$.
\end{center}
$\hfill\blacksquare$
\end{proof}

\textbf{Analysis of Subcase C2.}
For this subcase, we show that $\mathcal{M}$ can break the unforgeability of the MQV functionality proved in theorem \ref{the:mqv} if it succeeds with non-negligible probability. We build a simulator $\mathcal{S}_3$ which simulates $\mathcal{M}$'s environment. $\mathcal{S}_3$ takes as input a challenge $X_0$, and creates an experiment which includes $n$ honest parties and the attacker $\mathcal{M}$. All parties compute their keys normally. $\mathcal{S}_3$ randomly selects two parties $\hat{A}$ and $\hat{B}$, and randomly selects two integers $i,j\in[1,...,k]$. The simulation for $\mathcal{M}$'s environment proceeds as follows:
\begin{enumerate}\setcounter{enumi}{-1}
\item $\mathcal{S}_3$ models $\mathsf{ephem}_{P}()$ for all parties following the description in Section \ref{sec:formaldes}. $\mathcal{M}$ sets the types for all long-term keys. If the type of $A$ is not $\mathsf{TPM\_ALG\_ECMQV}$ and the type of $B$ is $\mathsf{TPM\_ALG\_ECMQV}$, $\mathcal{S}_3$ aborts. $\mathcal{S}_3$ can create oracles modeling the two-phase key exchange functionalities for each party normally as it possesses the long-term private keys of all parties. 
\item Initiate$(sc,\hat{P_1}, \hat{P_2})$: $\hat{P_1}$ executes the Initiate() activation of the protocol. However, if the session being created is the $i$-th session at $\hat{A}$, $\mathcal{S}_3$ checks whether $sc=\mathsf{TPM\_ALG\_ECMQV}$ and $\hat{P_2} = \hat{B}$. If so, $\mathcal{S}_3$ sets the ephemeral public key to be $X_0$, else $\mathcal{S}_3$ aborts. If the session being created is the $j$-th session at $\hat{B}$, $\mathcal{S}_3$ checks whether $sc\neq\mathsf{TPM\_ALG\_ECMQV}$ and $\hat{P_2}$ is $\hat{A}$. If so, $\mathcal{S}_3$ calls $\mathsf{ephem}_{B}()$ to create an ephemeral key, denoted by $Y_0$, and sets the outgoing ephemeral key of this session to be $Y_0$, else $\mathcal{S}_3$ aborts.
\item Respond$(sc,\hat{P_1}, \hat{P_2}, Y)$: $\hat{P_1}$ executes the Respond() activation of the protocol. However, if the session being created is the $i$-th session at $\hat{A}$, $\mathcal{S}_3$ checks whether  $sc=\mathsf{TPM\_ALG\_ECMQV}$, $\hat{P_2} = \hat{B}$, and $Y=Y_0$. If so, $\mathcal{S}_3$ provides $\mathcal{M}$ with the value $X_0$, else $\mathcal{S}_3$ aborts. If the session being created is the $j$-th session at $\hat{B}$, $\mathcal{S}_3$ checks whether $sc\neq\mathsf{TPM\_ALG\_ECMQV}$, $\hat{P_2} = \hat{A}$, and $Y=X_0$. If so, $\mathcal{S}_3$ calls $\mathsf{ephem}_{B}()$ to create an ephemeral key, denoted by $Y_0$, and sets the outgoing ephemeral key of this session to be $Y_0$, else $\mathcal{S}_3$ aborts.
\item Complete$(sc,\hat{P_1}, \hat{P_2}, X, Y)$: $\hat{P_1}$ executes the Complete() activation of the protocol. However, if the session is the $i$-th session at $\hat{A}$, $\mathcal{S}_3$ completes the session without computing a session key.
\item SessionStateReveal($s$): $\mathcal{S}_3$ returns to $\mathcal{M}$ the session state of session $s$. However, if $s$ is the $i$-th session at $\hat{A}$, $\mathcal{S}_3$ aborts.
\item SessionKeyReveal($s$): $\mathcal{S}_3$ returns to $\mathcal{M}$ the session key of $s$. If $s$ is the $i$-th session at $\hat{A}$, $\mathcal{S}_3$ aborts.
\item Corruption($\hat{P}$): $\mathcal{S}_3$ gives $\mathcal{M}$ the handle of the long-term key $P$. If $\mathcal{M}$ tries to corrupt $\hat{A}$ or $\hat{B}$, $\mathcal{S}_3$ aborts.
\item $H_1()$, $H_2()$, $avf()$, and $avf'()$ are simulated as random oracles in the usual way.
\end{enumerate}

\begin{proof}
The probability that $\mathcal{M}$ sets the type of $A$ to be $\mathsf{TPM\_ALG\_ECMQV}$ and the type of $B$ not to be $\mathsf{TPM\_ALG\_ECMQV}$ and selects the $i$-th session of $\hat{A}$ and the $j$-th session of $\hat{B}$ as the test session and its message-matching session is at least $\frac{1}{3}\times\frac{2}{3}\times\frac{1}{(nk)^2}=\frac{2}{9(nk)^2}$. Suppose that this is indeed the case: the type of $A$ is $\mathsf{TPM\_ALG\_ECMQV}$ and the type of $B$ is not $\mathsf{TPM\_ALG\_ECMQV}$, so $\mathcal{S}_3$ does not abort in Step 0; $\mathcal{M}$ is not allowed to corrupt $\hat{A}$ and $\hat{B}$, make SessionStateReveal and SessionKeyReveal queries to the $i$-th session of $\hat{A}$, so $\mathcal{S}_3$ does not abort in Step 1, 2, 4, 5, 6. Therefore, $\mathcal{S}_2$ simulates $\mathcal{M}$'s environment perfectly except with negligible probability. By the assumption, $\mathcal{M}$ correctly guesses the tuple $(Z=(Y_0B^e)^{h(x_0+da)},\hat{A},\hat{B})$ of the test session where $d=avf(X_0)$ and $e=avf(Y_0)$. We now show that $(Y_0,Z)$ is a valid forgery against $\mathcal{O}^{\mathsf{MQV}}_{B}$ on input $(X_0,A)$ where $X_0$ is the challenge:
\begin{enumerate}
\item $(Y_0,Z)$ is a valid return result of $\mathcal{O}^{\mathsf{MQV}}_{B}$ as $Z=(Y_0B^e)^{h(x_0+da)}=(X_0A^d)^{h(y_0+eb)}$.
\item $\mathcal{O}^{\mathsf{MQV}}_{B}$ never returns the result $(Y_0,Z)$ on input $(X_0,A)$ under $\mathcal{S}_3$: $\mathcal{O}^{\mathsf{MQV}}_{B}$ has never been created by $\mathcal{S}_3$ as the type of $B$ is not $\mathsf{TPM\_ALG\_ECMQV}$ in this subcase.
\item Since $\mathcal{M}$ is not allowed to corrupt $\hat{A}$ and $\hat{B}$, $\mathcal{M}$ does not know $a$ and $b$. So $\mathcal{M}$ does not know the private key pair $(a,x_0)$. Thus, $\mathcal{M}$ is under the constraint described in theorem \ref{the:mqv}.
\end{enumerate}
Finally we get: $Pr$($\mathcal{M}$ succeeds in forging $\mathcal{O}^{\mathsf{MQV}}_{B}$ under $\mathcal{S}_3$) $\geq$$\frac{2}{9(nk)^2}Pr(\mathcal{M})$.
$\hfill\blacksquare$
\end{proof}

\textbf{Analysis of Subcases C3 and C4.}
For the two subcases, we show that $\mathcal{M}$ can break the unforgeability of the MQV functionality proved in theorem \ref{the:mqv} if it succeeds with non-negligible probability. We build a simulator $\mathcal{S}_4$ which simulates $\mathcal{M}$'s environment. $\mathcal{S}_4$ takes as input a challenge $X_0$, and creates an experiment which includes $n$ honest parties and the attacker $\mathcal{M}$. All parties compute their keys normally. $\mathcal{S}_4$ randomly selects two parties $\hat{A}$ and $\hat{B}$, and randomly selects one integer $i\in[1,...,k]$. The simulation for $\mathcal{M}$'s environment proceeds as follows:
\begin{enumerate}\setcounter{enumi}{-1}
\item $\mathcal{S}_4$ models $\mathsf{ephem}_{P}()$ for all parties following the description in Section \ref{sec:formaldes}. $\mathcal{M}$ sets the type for all long-term keys. If the type of $A$ is not $\mathsf{TPM\_ALG\_ECMQV}$, $\mathcal{S}_4$ aborts. $\mathcal{S}_4$ can create oracles modeling the two-phase key exchange functionalities for each party normally as it possesses the long-term private keys of all the parties. 
\item Initiate$(sc,\hat{P_1}, \hat{P_2})$: $\hat{P_1}$ executes the Initiate() activation of the protocol. However, if the session being created is the $i$-th session at $\hat{A}$, $\mathcal{S}_4$ checks whether $sc=\mathsf{TPM\_ALG\_ECMQV}$ and $\hat{P_2} = \hat{B}$. If so, $\mathcal{S}_4$ sets the ephemeral public key to be $X_0$, else $\mathcal{S}_4$ aborts.
\item Respond$(sc,\hat{P_1}, \hat{P_2}, Y)$: $\hat{P_1}$ executes the Respond() activation of the protocol. However, if the session being created is the $i$-th session at $\hat{A}$, $\mathcal{S}_4$ checks whether  $sc=\mathsf{TPM\_ALG\_ECMQV}$ and $\hat{P_2} = \hat{B}$. If so, $\mathcal{S}_4$ provides $\mathcal{M}$ with the value $X_0$, else $\mathcal{S}_4$ aborts.
\item Complete$(sc,\hat{P_1}, \hat{P_2}, X, Y)$: $\hat{P_1}$ executes the Complete() activation of the protocol. However, if the session is the $i$-th session at $\hat{A}$, $\mathcal{S}_4$ completes the session without computing a session key.
\item SessionStateReveal($s$): $\mathcal{S}_4$ returns to $\mathcal{M}$ the session state of session $s$. However, if $s$ is the $i$-th session at $\hat{A}$, $\mathcal{S}_4$ aborts.
\item SessionKeyReveal($s$): $\mathcal{S}_4$ returns to $\mathcal{M}$ the session key of $s$. If $s$ is the $i$-th session at $\hat{A}$, $\mathcal{S}_4$ aborts.
\item Corruption($\hat{P}$): $\mathcal{S}_4$ gives $\mathcal{M}$ the handle of the long-term key $P$. If $\mathcal{M}$ tries to corrupt $\hat{A}$ or $\hat{B}$, $\mathcal{S}_4$ aborts.
\item $H_1()$, $H_2()$, $avf()$, and $avf'()$ are simulated as random oracles in the usual way.
\end{enumerate}

\begin{proof}
The probability that $\mathcal{M}$ sets the type of $A$ to be $\mathsf{TPM\_ALG\_ECMQV}$ and selects the $i$-th session of $\hat{A}$ as the test session is at least $\frac{1}{3}\times\frac{1}{n^2k}=\frac{1}{3n^2k}$. Suppose that this is indeed the case: the type of $A$ is $\mathsf{TPM\_ALG\_ECMQV}$, so $\mathcal{S}_4$ does not abort in Step 0; $\mathcal{M}$ is not allowed to corrupt $\hat{A}$ and $\hat{B}$, make SessionStateReveal and SessionKeyReveal queries to the $i$-th session of $\hat{A}$, so $\mathcal{S}_4$ does not abort in Step 1, 2, 4, 5, 6. Therefore, $\mathcal{S}_4$ simulates $\mathcal{M}$'s environment perfectly. By the assumption, $\mathcal{M}$ correctly guesses the tuple $(Z=(Y_0B^e)^{h(x_0+da)},\hat{A},\hat{B})$ of the test session where $d=avf(X_0)$ and $e=avf(Y_0)$. We now show that $(Y_0,Z)$ is a valid forgery against $\mathcal{O}^{\mathsf{MQV}}_{B}$ on input $(X_0,A)$ where $X_0$ is the challenge:
\begin{enumerate}
\item $(Y_0,Z)$ is a valid return result of $\mathcal{O}^{\mathsf{MQV}}_{B}$ as $Z=(Y_0B^e)^{h(x_0+da)}=(X_0A^d)^{h(y_0+eb)}$.
\item We now show that $\mathcal{O}^{\mathsf{MQV}}_{B}$ never returns the result $(Y_0,Z)$ on input $(X_0,A)$ under $\mathcal{S}_4$:
    \begin{enumerate}
    \item If the type of $B$ is $\mathsf{TPM\_ALG\_ECDH}$ or $\mathsf{TPM\_ALG\_SM2}$, $\mathcal{O}^{\mathsf{MQV}}_{B}$ has never been created by $\mathcal{S}_4$.
    \item If the type of $B$ is $\mathsf{TPM\_ALG\_ECMQV}$, $\mathcal{S}_4$ must create $\mathcal{O}^{\mathsf{MQV}}_{B}$ in step 0. However, if $\mathcal{O}^{\mathsf{MQV}}_{B}$ ever returned the result $(Y_0,Z)$ for some $Z$ on input $(A,X_0)$, $\hat{B}$ must have a session identified by $(sc=\mathsf{TPM\_ALG\_ECMQV}, \hat{B},\hat{A},Y_0,X_0)$, which is exactly the matching session of the test session. This contradicts that the test session has no matching session in these two subcases.
    \end{enumerate}
\item Since $\mathcal{M}$ is not allowed to corrupt $\hat{A}$ and $\hat{B}$, $\mathcal{M}$ does not know $a$ and $b$. So $\mathcal{M}$ does not know the private key pair $(a,x_0)$. Thus, $\mathcal{M}$ is under the constraint described in theorem \ref{the:mqv}.
\end{enumerate}
Finally we get: $Pr$($\mathcal{M}$ succeeds in forging $\mathcal{O}^{\mathsf{MQV}}_{B}$ under $\mathcal{S}_4$) $\geq$$\frac{1}{3n^2k}Pr(\mathcal{M})$.
$\hfill\blacksquare$
\end{proof}

\subsubsection{Analysis of Case 3}
The analysis of Case 3 is similar to Case 2. It is easy to get a full proof by following the analysis in Section \ref{subsec:case2}, so we omit the analysis of Case 3.

\subsection{Infeasibility of Key-replication Attacks}
If $\mathcal{M}$ successfully launches a key-replication attack against the test session $s=(sc,\hat{A},\hat{B},X_0,Y_0)$, he succeeds in establishing a session $s'=(sc',\hat{A}',\hat{B}',X',Y')$, which is different than $s$ and $(sc,\hat{B},\hat{A},Y_0,X_0)$ (the matching session of $s$) but has the same key as the test session $s$. The $sc$ of $s'$ must fall under one of the following three cases. By demonstrating that key-replication attacks are impossible in any of the three cases, we prove that $\mathcal{M}$ cannot launch key-replication attacks.
\begin{enumerate}
\item $sc=\mathsf{TPM\_ALG\_ECDH}$.
\item $sc=\mathsf{TPM\_ALG\_ECMQV}$.
\item $sc=\mathsf{TPM\_ALG\_SM2}$.
\end{enumerate}

\subsubsection{Analysis of Case 1}\label{subsubsec:key-rep-c1}
In this case, the session key of the test session is the value of the random oracle $H_1()$ on $\sigma=(Z_1,Z_2,\hat{A},\hat{B},X_0,Y_0)$. As the session key of the MQV or SM2 key exchange protocol is the value of the random oracle $H_2()$ on some 3-tuple $(Z,\hat{A},\hat{B})$, the session $s'$ must belong to a party whose long-term key is the type of $\mathsf{TPM\_ALG\_ECDH}$. So the session identifier of $s'$ must be $(sc,\hat{A},\hat{B},X_0,Y_0)$ or $(sc,\hat{B},\hat{A},Y_0,X_0)$ where $sc=\mathsf{TPM\_ALG\_ECDH}$, i.e., $s'$ is the test session or its matching session, which contradicts that $s'$ is different from $s$ and the matching session of $s$.
\subsubsection{Analysis of Cases 2 and 3}
We show that a key-replication attack is impossible by showing that a successful attacker would contradict the GDH assumption or break the unforgeability of MQV functionality or SM2 key exchange functionality.

Since $s'$ has the same session key as the test session $s$, $s'$ must have the same $\sigma$ as the test session. In all the subcases of Case 2 in Section \ref{subsec:case2}, all the simulators ($\mathcal{S}_2$, $\mathcal{S}_3$, and $\mathcal{S}_4$) provide $\mathcal{M}$ with the session state of all exposed sessions. Therefore, $\mathcal{M}$ can obtain the 3-tuple of $s$ by exposing $s'$ (this is allowed in the security model as $s'$ is not the matching session of $s$). This means that $\mathcal{M}$ is able to launch forging attacks. However, we have shown that if $\mathcal{M}$ succeeds in a forging attack: $\mathcal{S}_2$ would succeed in solving the GDH problem, and under $\mathcal{S}_3$ and $\mathcal{S}_4$, there would exist an attacker breaking the unforgeability of the MQV functionality of $\mathsf{tpm.KE}$. By applying the above argument and replacing the unforgeability of the MQV functionality of $\mathsf{tpm.KE}$ with the unforgeability of the SM2 key exchange functionality of $\mathsf{tpm.KE}$, we directly get the analysis of Case 3.

\section{Further Security Properties of $\mathsf{tpm.KE}$} \label{sec:further-tpm.ke}
Besides the basic security property defined by modern security models, it is desirable for AKE protocols to achieve the following two security properties. The key-compromise impersonation (KCI) resistance property: the knowledge of a party's long-term private key does not enable the attacker to impersonate \emph{other uncorrupted parties} to the party. The Perfect Forward Secrecy (PFS) property: the expired session keys established before the compromise of the long-term key cannot be recovered.

From the UKS attacks against MQV and SM2 key exchange protocols (\ref{app:kaliski} and \ref{app:xu}) we know that if the attacker is allowed to obtain the plaintext of the long-term key when he corrupts the corresponding party, these protocols will be unable to satisfy the basic security property defined by modern AKE security models. So the current version of $\mathsf{tpm.KE}$ can only achieve the security property defined by the security model in case that the attacker cannot obtain the plaintext of long-term keys (even he corrupts the parties).
Since $\mathsf{tpm.KE}$ can only be proven secure in situations where the plaintext of long-term keys cannot be revealed to the attacker, which does not satisfy the definitions of KCI-resistance and wPFS properties, it cannot provide the KCI-resistance and wPFS properties. But $\mathsf{tpm.KE}$ can satisfy weak forms of the two properties: (1) \emph{constrained KCI-resistance}; that is, the control of a party's long-term key handle does not enable the attacker to impersonate \emph{other uncontrolled parties} to the party; and (2) the \emph{constrained PFS property}; that is, the expired session keys established before the attacker controls the handle of the long-term key cannot be recovered. To prove the above weak forms of the two properties, all needed is to note that the proof of $\mathsf{tpm.KE}$ in Section \ref{sec:analysis} still holds if we allow the attacker to corrupt $\hat{A}$ and $\hat{B}$ which are the related parties of the test session, i.e., all the simulators do not abort when $\hat{A}$ and $\hat{B}$ are corrupted. The proof remains valid since the abort operations are never used in the proof.

\section{Suggestions on Usage of the Current Version of $\mathsf{tpm.KE}$} \label{sec:sug-tpm.ke}
Although we formally prove that $\mathsf{tpm.KE}$ satisfies the basic security property defined by modern AKE models, this is achieved under the following conditions: first, all the long-term keys must be generated by the TPM, and the attacker cannot register arbitrary keys; second, the attacker cannot obtain the plaintext of the long-term key even he corrupts the party; third, the attacker cannot obtain the unhashed value of the Full UM protocol. These constraints require that engineers should deploy $\mathsf{tpm.KE}$ properly, otherwise $\mathsf{tpm.KE}$ would be unable to provide secure communications. In order to help engineers to use $\mathsf{tpm.KE}$ securely, we give the following suggestions:
\begin{enumerate}
\item For a network that plans to use $\mathsf{tpm.KE}$ to protect its communications, we suggest that all devices in the network should be equipped with the TPM and the CA should only issue certificates for keys that are generated by the TPM. This requires the CA to check the validity of the TPM, and this can be done by leveraging the Privacy CA protocol \cite{privacyCA} or the direct anonymous attestation (DAA) protocol \cite{DAA} if higher anonymity is required.
\item The network administrator should guarantee that all TPM chips are well protected against sophisticated physical attacks which can obtain the secrets inside the TPM, and the administrator should know that if one TPM chip is physically broken, the whole network may no longer be secure.
\item The software running on the host which derives the session key from the return results of the TPM should be well protected. For example, run the software in the secure execution environment provided by the Intel SGX \cite{SGX-model-1,SGX-model-2} or ARM TrustZone \cite{tzwp} technologies and delete the return results of the TPM (especially $Z_1$ of the Full UM protocol) immediately after the session key is derived.
\end{enumerate}

\section{Revision of $\mathsf{tpm.KE}$} \label{sec:tpm.ke.rev}
A real-world network may contain various kinds of devices, and it is common that some devices are protected by the TPM and the others are not. Moreover, there exist physical attacks that can access all keys inside the TPM \cite{tpm-attack}. So it is hard for real-world networks to satisfy the conditions required to ensure $\mathsf{tpm.KE}$'s security. Therefore, the current version of $\mathsf{tpm.KE}$ puts a significant limitation on its application on real-world networks.

In order to make $\mathsf{tpm.KE}$ more applicable to real-world networks, we suggest revising TPM 2.0 specifications as follows: \emph{perform the session key derivation in the TPM rather than on the host, i.e., perform $H_1()$ and $H_2()$ in the TPM}. The revision only adds a hash computation to the TPM, which is negligible compared to the elliptic curve scalar multiplication. This revision is inspired by the HMQV protocol, the first proven secure implicitly AKE protocol, which requires that its unhashed value $Z$ should be stored in the same secure environment as the long-term key and only the session key is output outside of the environment. We give concrete suggestions on how to revise TPM 2.0 specifications: the only change is to revise TPM2\_ZGen\_2Phase() command in the TPM 2.0 Command specification \cite{TPM2.0Commands}. The command is revised to derive the session key inside the TPM and return the session key, while the original command returns the unhashed values $Z_1$ and $Z_2$. Figure \ref{fig:revision} describes the changes to the input and response of the TPM2\_ZGen\_2Phase() command.
\begin{itemize}
\item \textbf{Changes of input}: add the identities of the owner and peer of the session.
\item \textbf{Changes of response}: not return the unhashed values but the session key.
\end{itemize}

\begin{figure*}[htbp]
\begin{center}
\includegraphics[width=6.5in]{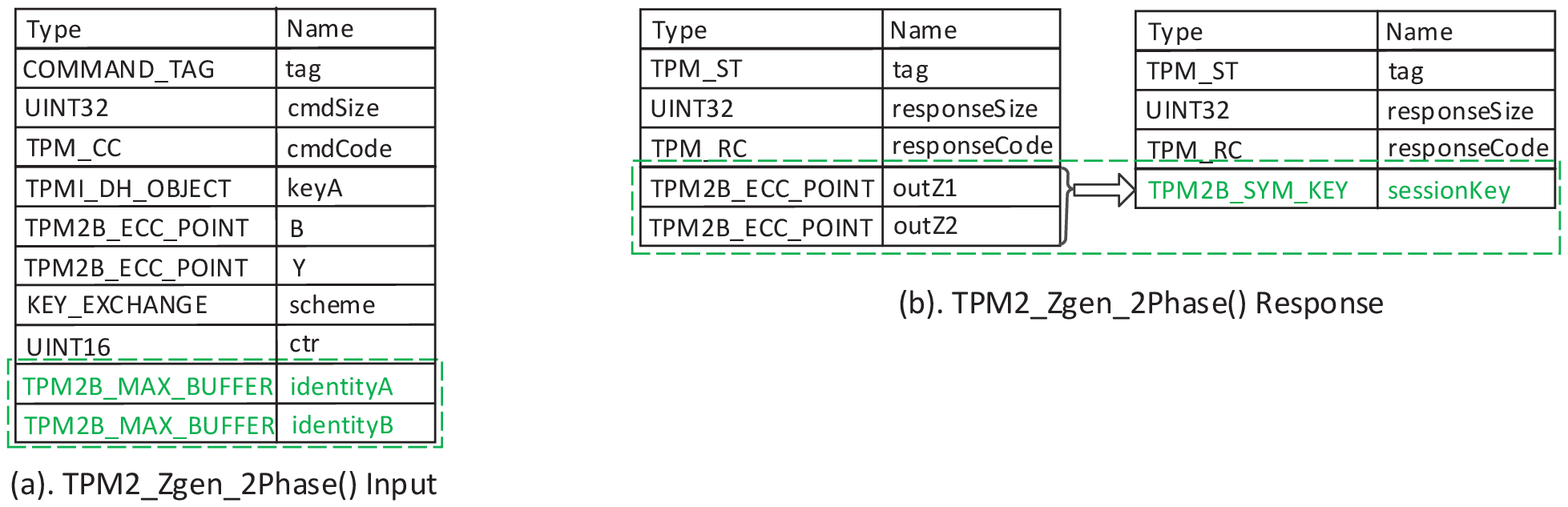}
\end{center}
\caption{\label{fig:revision}Revision of TPM2\_ZGen\_2Phase()}
\end{figure*}

\section{Preparation for the Security Analysis of the Revision} \label{sec:pre-tpm.ke.rev}
In this section, we do the preparation work for the security analysis of the revision of $\mathsf{tpm.KE}$ (denoted by $\mathsf{tpm.KE.rev}$), including the formal description of $\mathsf{tpm.KE.rev}$, the session state, and the attacker model.

\subsection{Formal Description of $\mathsf{tpm.KE.rev}$}
The first phase of $\mathsf{tpm.KE.rev}$ is the same as $\mathsf{tpm.KE}$, so we use the same $\mathsf{ephem()}_{A}$ to model the interface of the first phase of $\mathsf{tpm.KE.rev}$. As the TPM returns the session key in the second phase of $\mathsf{tpm.KE.rev}$, we modify the oracles $\mathcal{O}^{\mathsf{MQV}}_{A}$ and $\mathcal{O}^{\mathsf{SM2}}_{A}$ to return the session key. We get the formal description of $\mathsf{tpm.KE.rev}$ by using the modified oracles $\mathcal{O}^{\mathsf{MQV}}_{A}$ and $\mathcal{O}^{\mathsf{SM2}}_{A}$ to replace the original oracles of the formal description of $\mathsf{tpm.KE}$ in Section \ref{sec:formaldes}.

As our formal analysis of $\mathsf{tpm.KE.rev}$ contains parties that are not equipped with the TPM, we need to describe how the protocols are implemented in these parties. These parties also use the \emph{Initiate}, \emph{Respond} and \emph{Complete} activations to construct protocol sessions. The three activations for these parties differ from the activations for $\mathsf{tpm.KE.rev}$ in that they do not leverage the interfaces of the TPM, i.e., $\mathsf{ephem()}_{A}$, $\mathcal{O}^{\mathsf{EC}}_{A}$, $\mathcal{O}^{\mathsf{MQV}}_{A}$, and $\mathcal{O}^{\mathsf{SM2}}_{A}$, but run following the specifications of the protocols.

\subsection{Session State}

The session state includes the information returned by the TPM. As the ephemeral key and unhashed values (such as $Z_1$ and $Z_2$) are stored inside the TPM, we specify that the session state only includes the session key.

\subsection{Attacker Model for $\mathsf{tpm.KE.rev}$}
The attacker model for $\mathsf{tpm.KE.rev}$ is used to capture realistic attack capabilities in real-world networks, which are much stronger than the attack capabilities captured by the model for $\mathsf{tpm.KE}$. In particular, the model for $\mathsf{tpm.KE.rev}$ allows the attacker to register arbitrary keys to the CA and obtain the plaintext of a long-term key by corrupting the party. To capture the above stronger attack capabilities, we add the \emph{EstablishParty} query to the model for $\mathsf{tpm.KE}$ (Section \ref{sec:model}) and modify the \emph{SessionStateReveal} and \emph{Corruption} queries.

\begin{itemize}
\item \textbf{EstablishParty($\hat{P}$)}: this query is newly added to the model to allow the attacker to register a static public key on behalf of party $\hat{P}$.
\item \textbf{SessionStateReveal(s)}: $\mathcal{M}$ directly queries at session $s$ while still incomplete and learns the session state for $s$. The session state only includes the session key, while the session state in the model for $\mathsf{tpm.KE}$ includes the unhashed values.
\item \textbf{Corruption($\hat{P}$)}: this query allows $\mathcal{M}$ to learn the plaintext of the long-term private key of party $\hat{P}$, while this query in the model for $\mathsf{tpm.KE}$ only allows the attacker to obtain the black-box access of the long-term key.
\end{itemize}

\section{Security Analysis of $\mathsf{tpm.KE.rev}$}\label{sec:analysis-rev}
In order to make our analysis be close to real-world networks as much as possible, we do not make any assumption whether the parties in the network are equipped with the TPM, including the owner and the peer of the test session.

\begin{theorem}\label{the-main-rev}
Under the GDH assumption, with hash functions $H_1()$, $H_2()$, $avf()$, and $avf'()$ modeled as random oracles, $\mathsf{tpm.KE.rev}$ is secure in the unified model.
\end{theorem}

The proof of the above theorem follows from the definition of secure key exchange protocols and the following two lemmas.

\begin{lemma}\label{lemma:same-rev}
If two parties $\hat{A}$ and $\hat{B}$ complete matching sessions, their session keys are the same.
\end{lemma}

\begin{lemma}\label{lemma:core-rev}
Under the GDH assumption, there is no feasible attacker that succeeds in distinguishing the session key of an unexposed session with non-negligible probability.
\end{lemma}

The proof of Lemma \ref{lemma:same-rev} is the same as Lemma \ref{lemma:same}, and we focus on the proof of Lemma \ref{lemma:core-rev}. Let $\mathcal{M}$ be any attacker against $\mathsf{tpm.KE.rev}$. $\mathcal{M}$ can launch the forging attack or key-replication attack to distinguish the session key of the test session from a random value. We will show that if either of the attacks succeed with non-negligible probability, there exists a solver succeeding in solving the GDH problem with non-negligible probability.

\subsection{Infeasibility of Forging Attacks}
Let $(sc,\hat{A},\hat{B},X_0,Y_0)$ be the test session for whose tuple $\mathcal{M}$ outputs a correct guess. $sc$ can fall under one of the following three cases:
\begin{enumerate}
\item $sc=\mathsf{TPM\_ALG\_ECDH}$.
\item $sc=\mathsf{TPM\_ALG\_ECMQV}$.
\item $sc=\mathsf{TPM\_ALG\_SM2}$.
\end{enumerate}

As $\mathcal{M}$ succeeds with non-negligible probability in the forging attack, there is at least one of the above three cases that occurs with non-negligible probability. We assume that $\mathcal{M}$ operates in an environment that involves at most $n$ parties and each party participates in at most $k$ sessions. We analyze the three cases separately in the following.

\subsubsection{Analysis of Case 1}
For this case, we build a GDH solver $\mathcal{S}_5$. It takes as input a pair $(A,B)$, creates an experiment which includes $n$ honest parties and the attacker $\mathcal{M}$, and is given access to a DDH oracle $DDH$. $\mathcal{S}_5$ randomly decides whether each party is protected by the TPM or not, randomly selects two parties $\hat{A}$ and $\hat{B}$ from the honest parties and sets their public keys to be $A$ and $B$ respectively. All the other parties compute their keys normally. Furthermore, $\mathcal{S}_5$ randomly selects an integer $i\in[1,...,k]$. The simulation for $\mathcal{M}$'s environment proceeds as follows.
\begin{enumerate}\setcounter{enumi}{-1}
\item $\mathcal{S}_5$ models $\mathsf{ephem}_{P}()$ following the description in Section \ref{sec:formaldes} for all parties protected by the TPM. $\mathcal{M}$ sets the types of all long-term keys. If the type of $A$ is not $\mathsf{TPM\_ALG\_ECDH}$, $\mathcal{S}_5$ aborts. $\mathcal{S}_5$ creates oracles modeling the two-phase key exchange functionalities for parties protected by the TPM normally except parties $\hat{A}$ and $\hat{B}$. The procedures of $H_1()$ and $H_2()$ are described below. $avf()$ and $avf'()$ are simulated as random oracles in the usual way.
\item Initiate$(sc,\hat{P_1}, \hat{P_2})$: $\hat{P_1}$ executes the Initiate() activation of the protocol. However, if the session being created is the $i$-th session at $\hat{A}$, $\mathcal{S}_5$ checks whether $sc=\mathsf{TPM\_ALG\_ECDH}$ and $\hat{P_2} = \hat{B}$. If not, $\mathcal{S}_5$ aborts.
\item Respond$(sc,\hat{P_1}, \hat{P_2}, Y)$: With the exception of $\hat{A}$ and $\hat{B}$ (whose behaviors we explain below), $\hat{P_1}$ executes the Respond() activation of the protocol. However, if the session being created is the $i$-th session at $\hat{A}$, $\mathcal{S}_5$ checks whether  $sc=\mathsf{TPM\_ALG\_ECDH}$ and $\hat{P_2} = \hat{B}$. If not, $\mathcal{S}_5$ aborts.
\item Complete$(sc,\hat{P_1}, \hat{P_2}, X, Y)$:  With the exception of $\hat{A}$ and $\hat{B}$ (whose behaviors we explain below), $\hat{P_1}$ executes the Complete() activation of the protocol. However, if the session is the $i$-th session at $\hat{A}$, $\mathcal{S}_5$ completes the session without computing a session key.
\item Establish$(\hat{P}, P)$: $\mathcal{S}_5$ creates the party $\hat{P}$ whose public key is $P$ and marks $\hat{P}$ as corrupted.
\item If $\hat{A}$ is protected by the TPM, with the input $(sc, keyA, ctr_x, P, Y)$, $\mathcal{S}_5$ creates the oracle $\mathcal{O}^{\mathsf{EC}}_{A}$ as follows: compute the session key $k=H_{spec}(\hat{A},\hat{P},X,Y)$. $H_{spec}()$ is simulated as a random oracle. \label{step:modelECA-TPM}
\item Respond$(sc,\hat{A}, \hat{P}, Y)$:
    \begin{enumerate}
    \item If $\hat{A}$ is protected by the TPM, $\mathcal{S}_5$ simulates this session activation with the help of $\mathsf{ephem}_{A}()$ and $\mathcal{O}^{\mathsf{EC}}_{A}$.
    \item If $\hat{A}$ is not protected by the TPM, $\mathcal{S}_5$ generates an ephemeral key pair $(x, X=g^x)$, outputs $X$, and computes the session key $k=H_{spec}(\hat{A},\hat{B},X,Y)$.
    \end{enumerate}
\item Complete$(sc,\hat{A}, \hat{P}, X, Y)$:
    \begin{enumerate}
    \item If $\hat{A}$ is protected by the TPM, $\mathcal{S}_5$ simulates this session activation with the help of $\mathsf{ephem}_{A}()$ and $\mathcal{O}^{\mathsf{EC}}_{A}$.
    \item If $\hat{A}$ is not protected by the TPM, $\mathcal{S}_5$ computes the session key $k=H_{spec}(\hat{A},\hat{P},X,Y)$.
    \end{enumerate}

\item If $\hat{B}$ is protected by the TPM, with the input $(sc, keyB, ctr_y, P, X)$, $\mathcal{S}_5$ creates the oracle modeling the two-phase key exchange functionality for party $\hat{B}$: recover the ephemeral key pair $(y,Y=g^y)$ by $ctr_y$, and compute the session key $k=H_{spec}(\hat{B},\hat{P},Y,X)$. \label{step:oracleB-c1}
\item Respond$(sc,\hat{B}, \hat{P}, X)$:
    \begin{enumerate}
    \item If $\hat{B}$ is protected by the TPM, $\mathcal{S}_5$ simulates this session activation with the help of $\mathsf{ephem}_{B}()$ and the oracle for $\hat{B}$ (Step \ref{step:oracleB-c1}).
    \item If $\hat{B}$ is not protected by the TPM, $\mathcal{S}_5$ simulates this session activation as follows: generate an ephemeral key pair $(y, Y=g^y)$, output $Y$, and compute the session key $k=H_{spec}(\hat{B},\hat{P},Y,X)$.
    \end{enumerate}
\item Complete$(sc,\hat{B}, \hat{P}, Y, X)$:
    \begin{enumerate}
    \item If $\hat{B}$ is protected by the TPM, $\mathcal{S}_5$ simulates this session activation with the help of $\mathsf{ephem}_{B}()$ and the oracle for $\hat{B}$ (Step \ref{step:oracleB-c1}).
    \item If $\hat{B}$ is not protected by the TPM, $\mathcal{S}_5$ sets the session key to be $k=H_{spec}(\hat{B},\hat{P},Y,X)$.
    \end{enumerate}

\item SessionStateReveal($s$): $\mathcal{S}_5$ returns to $\mathcal{M}$ the session state of session $s$. However, if $s$ is the $i$-th session at $\hat{A}$, $\mathcal{S}_5$ aborts.
\item SessionKeyReveal($s$): $\mathcal{S}_5$ returns to $\mathcal{M}$ the session key of $s$. If $s$ is the $i$-th session at $\hat{A}$, $\mathcal{S}_5$ aborts.
\item Corruption($\hat{P}$): $\mathcal{S}_5$ gives $\mathcal{M}$ the long-term key pair of $\hat{P}$. If $\mathcal{M}$ tries to corrupt $\hat{A}$ or $\hat{B}$, $\mathcal{S}_5$ aborts.
\item $H_1(\sigma)$ function for some $\sigma=(Z_1,Z_2,\hat{P_1},\hat{P_2},X,Y)$ proceeds as follows:
    \begin{enumerate}
    \item If $\hat{P_1}=\hat{A}$, $\hat{P_2}=\hat{B}$, and $DDH(A,B,Z_1)=1$, $\mathcal{S}_5$ aborts $\mathcal{M}$ and succeeds by outputting CDH$(A,B)=Z_1$.
    \item If the value of the function on input $\sigma$ has been defined previously, return it.
    \item If the value of $H_{spec}()$ on input $(\hat{P_1},\hat{P_2},X,Y)$ has been defined previously, return it.
    \item Pick a random key $k\in_R\{0,1\}^{\lambda}$, and define $H_1(\sigma)=k$.
    \end{enumerate}
\item $H_2(\sigma)$ function for some $\sigma=(Z,\hat{P_1},\hat{P_2})$ proceeds as follows:
    \begin{enumerate}
    \item If the value of the function on input $\sigma$ has been defined previously, return it.
    \item If not defined, go over all the previous calls to $H_{spec}()$ and for each previous call of the form $H_{spec}(\hat{P_1},\hat{P_2},X,Y)=v$ proceed as follows according to the type of $P_1$.
        \begin{enumerate}
        \item Case $\mathsf{TPM\_ALG\_ECMQV}$: check if $DDH(XP_1^d,YP_2^e,Z^{1/h})=1$ where $d=avf(X)$ and $e=avf(Y)$; if the condition holds, return $v$.
        \item Case $\mathsf{TPM\_ALG\_SM2}$: check if $DDH(P_1X^d,P_2Y^e,Z^{1/h})=1$ where $d=avf'(X)$ and $e=avf'(Y)$; if the condition holds, return $v$.
        \end{enumerate}
    \item If no previous calls of that form are found, pick a random key $w$ and define $H_2(Z,\hat{P_1},\hat{P_2})=w$.
    \end{enumerate}
\item $H_{spec}(\hat{P_1},\hat{P_2},X,Y)$ function proceeds as follows:
    \begin{enumerate}
    \item If the value of the function on that input has been defined previously, return it.
    \item If not defined, proceed as follows according to the type of $P_1$.
        \begin{enumerate}
        \item Case $\mathsf{TPM\_ALG\_ECDH}$: go over all the previous calls to $H_1()$ and for each previous call of the form $H_1(Z_1,Z_2,\hat{P_1},\hat{P_2},X,Y)=v$ check if $DDH(P_1,P_2,Z_1)=1$; if the condition holds, return $v$.
        \item Case $\mathsf{TPM\_ALG\_ECMQV}$: go over all the previous calls to $H_2()$ and for each previous call of the form $H_2(Z,\hat{P_1},\hat{P_2})=v$ check if $DDH(XP_1^d,YP_2^e,Z^{1/h})=1$ where $d=avf(X)$ and $e=avf(Y)$; if the condition holds, return $v$.
        \item Case $\mathsf{TPM\_ALG\_SM2}$: go over all the previous calls to $H_2()$ and for each previous call of the form $H_2(Z,\hat{P_1},\hat{P_2})=v$ check if $DDH(P_1X^d,P_2Y^e,Z^{1/h})=1$ where $d=avf'(X)$ and $e=avf'(Y)$; if the condition holds, return $v$.
        \end{enumerate}
    \item If not found, pick a random key $w$, define $H_{spec}(\hat{P_1},\hat{P_2},X,Y)=w$.
    \end{enumerate}
\end{enumerate}

\begin{proof}
The probability that $\mathcal{M}$ sets the type of $A$ to be $\mathsf{TPM\_ALG\_ECDH}$ and selects the $i$-th session at $\hat{A}$ and the peer of the test session is party $\hat{B}$ is at least $\frac{1}{3n^2k}$. Suppose that this indeed the case: the type of $A$ is $\mathsf{TPM\_ALG\_ECDH}$, so $\mathcal{S}_5$ does not abort in Step 0; $\mathcal{M}$ is not allowed to corrupt $\hat{A}$ and $\hat{B}$, make SessionStateReveal and SessionKeyReveal queries to the $i$-th session at $\hat{A}$, so $\mathcal{S}_5$ does not abort in Step 1, 2, 11, 12, 13. Therefore, $\mathcal{S}_5$ simulates $\mathcal{M}$'s environment perfectly. Thus, if $\mathcal{M}$ wins with non-negligible probability in this case, the success probability of $\mathcal{S}_5$ is bounded by
\begin{center}
$Pr(\mathcal{S}_5)\geq \frac{1}{3n^2k}Pr(\mathcal{M})$.
\end{center}
$\hfill\blacksquare$
\end{proof}

\subsubsection{Analysis of Case 2}\label{subsec:case2-rev}
We build a GDH solver $\mathcal{S}_6$ for this case. $\mathcal{S}_6$ takes as input a pair $(X_0,B)$, creates an experiment which includes $n$ honest parties and the attacker $\mathcal{M}$, and is given access to a DDH oracle $DDH$. $\mathcal{S}_6$ randomly decides whether each party is protected by the TPM or not, randomly selects a party and sets its public key to be $B$. All the other parties compute their keys normally. Furthermore, $\mathcal{S}_6$ randomly selects an integer $i\in[1,...,k]$. The simulation for $\mathcal{M}$'s environment proceeds as follows.
\begin{enumerate}\setcounter{enumi}{-1}
\item $\mathcal{S}_6$ models $\mathsf{ephem}_{P}()$ following the description in Section \ref{sec:formaldes} for all parties protected by the TPM. $\mathcal{M}$ sets the types of all long-term keys. If the type of $B$ is not $\mathsf{TPM\_ALG\_ECMQV}$, $\mathcal{S}_6$ aborts. $\mathcal{S}_6$ creates oracles modeling the two-phase key exchange functionalities for parties protected by the TPM normally except $\hat{B}$. The procedures of $H_1()$ and $H_2()$ are described below. $avf()$ and $avf'()$ are simulated as random oracles in the usual way.
\item Initiate$(sc,\hat{P_1}, \hat{P_2})$: $\hat{P_1}$ executes the Initiate() activation of the protocol. However, if the session being created is the $i$-th session at $\hat{A}$, $\mathcal{S}_6$ checks whether $sc=\mathsf{TPM\_ALG\_ECMQV}$ and $\hat{P_2} = \hat{B}$. If so, $\mathcal{S}_6$ sets the ephemeral public key to be $X_0$, else $\mathcal{S}_6$ aborts.
\item Respond$(sc,\hat{P_1}, \hat{P_2}, Y)$: With the exception of $\hat{B}$ (whose behaviors we explain below), $\hat{P_1}$ executes the Respond() activation of the protocol. However, if the session being created is the $i$-th session at $\hat{A}$, $\mathcal{S}_6$ checks whether $sc=\mathsf{TPM\_ALG\_ECMQV}$ and $\hat{P_2} = \hat{B}$. If so, $\mathcal{S}_6$ sets the ephemeral public key to be $X_0$ and does not compute the session key, else $\mathcal{S}_6$ aborts.
\item Complete$(sc,\hat{P_1}, \hat{P_2}, X, Y)$: With the exception of $\hat{B}$, $\hat{P_1}$ executes the Complete() activation of the protocol. However, if the session is the $i$-th session at $\hat{A}$, $\mathcal{S}_6$ completes the session without computing a session key.
\item Establish$(\hat{P},P)$: $\mathcal{S}_6$ creates the party $\hat{P}$ whose public key is $P$, and marks $\hat{P}$ as corrupted.

\item If $\hat{B}$ is protected by the TPM, with the input $(sc, keyB, ctr_y, P, X)$, $\mathcal{S}_6$ creates the oracle modeling the two-phase key exchange functionality for party $\hat{B}$: recover the ephemeral key pair $(y,Y=g^y)$ by $ctr_y$, and compute the session key $k=H_{spec}(\hat{B},\hat{P},Y,X)$. \label{step:oracleB-c2}
\item Respond$(sc,\hat{B}, \hat{P}, X)$:
    \begin{enumerate}
    \item If $\hat{B}$ is protected by the TPM, $\mathcal{S}_6$ simulates this session activation with the help of $\mathsf{ephem}_{B}()$ and the oracle for $\hat{B}$ (Step \ref{step:oracleB-c2}).
    \item If $\hat{B}$ is not protected by the TPM, $\mathcal{S}_6$ simulates this session activation as follows: generate an ephemeral key pair $(y, Y=g^y)$, output $Y$, and compute the session key $k=H_{spec}(\hat{B},\hat{P},Y,X)$.
    \end{enumerate}
\item Complete$(sc,\hat{B}, \hat{P}, Y, X)$:
    \begin{enumerate}
    \item If $\hat{B}$ is protected by the TPM, $\mathcal{S}_6$ simulates this session activation with the help of $\mathsf{ephem}_{B}()$ and the oracle for $\hat{B}$ (Step \ref{step:oracleB-c2}).
    \item If $\hat{B}$ is not protected by the TPM, $\mathcal{S}_6$ sets the session key to be $k=H_{spec}(\hat{B},\hat{P},Y,X)$.
    \end{enumerate}

\item SessionStateReveal($s$): $\mathcal{S}_6$ returns to $\mathcal{M}$ the session state of session $s$. However, if $s$ is the $i$-th session at $\hat{A}$ or the matching session of it (if such a session exists), $\mathcal{S}_6$ aborts.
\item SessionKeyReveal($s$): $\mathcal{S}_6$ returns to $\mathcal{M}$ the session key of $s$. If $s$ is the $i$-th session at $\hat{A}$ or the matching session of it (if such a session exists), $\mathcal{S}_6$ aborts.
\item Corruption($\hat{P}$): $\mathcal{S}_6$ gives $\mathcal{M}$ the long-term key pair of $\hat{P}$. If $\mathcal{M}$ tries to corrupt $\hat{A}$ or $\hat{B}$, $\mathcal{S}_6$ aborts.

\item $H_1(\sigma)$ function for some $\sigma=(Z_1,Z_2,\hat{P_1},\hat{P_2},X,Y)$ proceeds as follows:
    \begin{enumerate}
    \item If the value of the function on input $\sigma$ has been defined previously, return it.
    \item If the value of $H_{spec}()$ on input $(\hat{P_1},\hat{P_2},X,Y)$ has been defined previously, return it.
    \item Pick a random key $k\in_R\{0,1\}^{\lambda}$, and define $H_1(\sigma)=k$.
    \end{enumerate}
\item $H_2(\sigma)$ function for some $\sigma=(Z,\hat{P_1},\hat{P_2})$ proceeds as follows:
    \begin{enumerate}
    \item If the value of the function on input $\sigma$ has been defined previously, return it.
    \item If not defined, go over all the previous calls to $H_{spec}()$ and for each previous call of the form $H_{spec}(\hat{P_1},\hat{P_2},X,Y)=v$ proceed as follows according to the type of $P_1$.
        \begin{enumerate}
        \item Case $\mathsf{TPM\_ALG\_ECMQV}$: check if $DDH(XP_1^d,YP_2^e,Z^{1/h})=1$ where $d=avf(X)$ and $e=avf(Y)$; if the condition holds, return $v$.
        \item Case $\mathsf{TPM\_ALG\_SM2}$: check if $DDH(P_1X^d,P_2Y^e,Z^{1/h})=1$ where $d=avf'(X)$ and $e=avf'(Y)$; if the condition holds, return $v$.
        \end{enumerate}
    \item If not found, pick a random key $w$ and define $H_2(Z,\hat{P_1},\hat{P_2})=w$.
    \end{enumerate}
\item $H_{spec}(\hat{P_1},\hat{P_2},X,Y)$ function proceeds as follows:
    \begin{enumerate}
    \item If the value of the function on that input has been defined previously, return it.
    \item If not defined, proceed as follows according to the type of $P_1$.
        \begin{enumerate}
        \item Case $\mathsf{TPM\_ALG\_ECDH}$: go over all the previous calls to $H_1()$ and for each previous call of the form $H_1(Z_1,Z_2,\hat{P_1},\hat{P_2},X,Y)=v$ check if $DDH(P_1,P_2,Z_1)=1$; if the condition holds, return $v$.
        \item Case $\mathsf{TPM\_ALG\_ECMQV}$: go over all the previous calls to $H_2()$ and for each previous call of the form $H_2(Z,\hat{P_1},\hat{P_2})=v$ check if $DDH(XP_1^d,YP_2^e,Z^{1/h})=1$ where $d=avf(X)$ and $e=avf(Y)$; if the condition holds, return $v$.
        \item Case $\mathsf{TPM\_ALG\_SM2}$: go over all the previous calls to $H_2()$ and for each previous call of the form $H_2(Z,\hat{P_1},\hat{P_2})=v$ check if $DDH(P_1X^d,P_2Y^e,Z^{1/h})=1$ where $d=avf'(X)$ and $e=avf'(Y)$; if the condition holds, return $v$.
        \end{enumerate}
    \item If not found, pick a random key $w$ and define $H_{spec}(\hat{P_1},\hat{P_2},X,Y)=w$.
    \end{enumerate}
\end{enumerate}

\begin{proof}
The probability that $\mathcal{M}$ sets the type of $A$ to be $\mathsf{TPM\_ALG\_ECMQV}$ and selects the $i$-th session at $\hat{A}$ as the test session is at least $\frac{1}{3}\times\frac{1}{n^2k}=\frac{1}{3n^2k}$. Suppose that this is indeed the case: the type of $A$ is $\mathsf{TPM\_ALG\_ECMQV}$, so $\mathcal{S}_6$ does not abort in Step 0; $\mathcal{M}$ is not allowed to corrupt $\hat{A}$ and $\hat{B}$, make SessionStateReveal and SessionKeyReveal queries to the $i$-th session at $\hat{A}$ or its matching session (if such a session exists), so $\mathcal{S}_6$ does not abort in Step 1, 2, 8, 9, 10. Therefore, $\mathcal{S}_6$ simulates $\mathcal{M}$'s environment perfectly.
If $\mathcal{M}$ wins the forging attack, he computes the 3-tuple $(Z,\hat{A},\hat{B})$. Without the knowledge of the private key $y$ of $Y$, $\mathcal{S}_6$ is unable to compute $CDH(X_0,B)$. $\mathcal{S}_6$ computes $T=\frac{Z^{1/h}}{B^aY_0^{ae}}$ where $d=avf(X_0)$ and $e=avf(Y_0)$.

Following the Forking Lemma \cite{fork96} approach, $\mathcal{S}_6$ runs $\mathcal{M}$ on the same input and same coin flips but with carefully modified answers to the $avf()$ queries. Note that $\mathcal{M}$ must have queried $avf(Y_0)$ in its first run because otherwise $\mathcal{M}$ would be unable to compute $Z$ of the test session. For the second run of $\mathcal{M}$, $\mathcal{S}_6$ responds to $avf(Y_0)$ with a value $e'\neq e$ selected uniformly at random. If $\mathcal{M}$ succeeds in the second run, $\mathcal{S}_6$ computes $T'=\frac{Z^{1/h}}{B^aY_0^{ae'}}$. And thereafter $\mathcal{S}_6$ computes $CDH(X_0,B)=(T'/T)^{e-e'}$. Thus, if $\mathcal{M}$ wins with non-negligible probability in this case, the success probability of $\mathcal{S}_6$ is bounded by:
\begin{center}
$Pr(\mathcal{S}_6)\geq \frac{C}{3n^2k}Pr(\mathcal{M})$
\end{center}
where $C$ is a constant arising from the use of the Forking Lemma.
$\hfill\blacksquare$
\end{proof}

\subsubsection{Analysis of Case 3}
The analysis of Case 3 is similar to Case 2. It is easy to get a full proof by following the analysis in Section \ref{subsec:case2-rev}, so we omit the analysis of Case 3.

\subsection{Infeasibility of Key-replication Attacks}

If $\mathcal{M}$ successfully launches a key-replication attack against the test session $s=(sc,\hat{A},\hat{B},X_0,Y_0)$, he succeeds in establishing a session $s'=(sc',\hat{A}',\hat{B}',X',Y')$, which is different than $s$ and $(sc,\hat{B},\hat{A},Y_0,X_0)$ (the matching session of $s$) but has the same key as the test session $s$. The $sc$ of $s'$ must fall under one of the following three cases. By demonstrating that key-replication attacks are impossible in any of the three cases, we prove that $\mathcal{M}$ cannot launch key-replication attacks.
\begin{enumerate}
\item $sc=\mathsf{TPM\_ALG\_ECDH}$.
\item $sc=\mathsf{TPM\_ALG\_ECMQV}$.
\item $sc=\mathsf{TPM\_ALG\_SM2}$.
\end{enumerate}

\subsubsection{Analysis of Case 1}
The analysis of this case is the same as the analysis in Section \ref{subsubsec:key-rep-c1} for $\mathsf{tpm.KE}$.

\subsubsection{Analysis of Case 2}
We prove that key-replication attacks are impossible in this case by showing that a successful key-replication attack would allow the attacker to launch forging attacks, contradicting the GDH assumption.
We build a GDH solver $\mathcal{S}_7$ for this case. $\mathcal{S}_7$ takes as input a pair $(X_0,B)$, simulates a similar environment for $\mathcal{M}$ as $\mathcal{S}_6$ does (Section \ref{subsec:case2-rev}) except the simulation of session activations at $\hat{B}$ and the \emph{SessionStateReveal} query.
\begin{itemize}
\item If $\hat{B}$ is protected by the TPM, $\mathcal{S}_7$ models $\mathsf{ephem}_{B}()$ and $\mathcal{O}^{\mathsf{MQV}}_B$ as follows:
    \begin{itemize}
    \item $\mathcal{S}_7$ creates a Table $T$ and models $\mathsf{ephem}_{B}()$ as follows.
        \begin{enumerate}
        \item Randomly choose $e,s\in Z_q$.
        \item Set $Y=g^s/B^e$ and $e=avf(Y)$.
        \item Randomly choose an index $ctr$, and add a record $(ctr, e,s,Y,-)$ to $T$.
        \item Return $ctr$ and $Y$.
        \end{enumerate}
    \item With the input $(sc, keyB, ctr_y, P, X)$, $\mathcal{S}_7$ models $\mathcal{O}^{\mathsf{MQV}}_B$ as follows:
        \begin{enumerate}
        \item Check whether (1) $sc=\mathsf{TPM\_ALG\_ECMQV}$, (2) $P$ and $X$ are on the curve associated with $B$, and (3) the last element of the record in $T$ indexed by $ctr_y$ is `$-$'. If the above checks succeed, continue, else return an error.
        \item Suppose the entry in $T$ indexed by $ctr_y$ is $(ctr_y, e,s,Y,-)$, set $Z_1=(XP^d)^{hs}$ where $d=avf(X)$, and set the last element of the entry to be $\times$.
        \item Compute the session key $k=H_2(Z_1,\hat{B},\hat{P})$, and return $k$.
        \end{enumerate}
    \end{itemize}

\item Initiate$(sc,\hat{B},\hat{P})$:
    \begin{itemize}
    \item If $\hat{B}$ is protected by the TPM, $\mathcal{S}_7$ simulates this session activation with the help of $\mathsf{ephem}_{B}()$.
    \item If $\hat{B}$ is not protected by the TPM, $\mathcal{S}_7$ simulates this session activation as follows with the help of a Table $T'$:
        \begin{enumerate}
        \item Randomly choose $e,s\in Z_q$.
        \item Set $Y=g^s/B^e$ and $e=avf(Y)$.
        \item Add an entry $(e,s,Y,-)$ to $T'$.
        \item Output $Y$.
        \end{enumerate}
    \end{itemize}

\item Respond$(sc,\hat{B}, \hat{P}, X)$:
    \begin{enumerate}
    \item If $\hat{B}$ is protected by the TPM, $\mathcal{S}_7$ simulates this session activation with the help of $\mathsf{ephem}_{B}()$ and $\mathcal{O}^{\mathsf{MQV}}_B$.
    \item If $\hat{B}$ is not protected by the TPM, $\mathcal{S}_7$ simulates this session activation as follows:
        \begin{enumerate}
        \item Randomly choose $e,s\in Z_q$.
        \item Set $Y=g^s/B^e$ and $e=avf(Y)$.
        \item Compute $Z_1=(XP^d)^{hs}$ where $d=avf(X)$.
        \item Output $Y$, and compute the session key $k=H_{2}(Z_1,\hat{B},\hat{P})$.
        \end{enumerate}
    \end{enumerate}
\item Complete$(sc,\hat{B}, \hat{P}, Y, X)$:
    \begin{enumerate}
    \item If $\hat{B}$ is protected by the TPM, $\mathcal{S}_7$ simulates this session activation with the help of $\mathsf{ephem}_{B}()$ and $\mathcal{O}^{\mathsf{MQV}}_B$.
    \item If $\hat{B}$ is not protected by the TPM, $\mathcal{S}_7$ proceeds as follows:
        \begin{enumerate}
        \item Check whether (1) $sc=\mathsf{TPM\_ALG\_ECMQV}$, (2) $P$ and $X$ are on the curve associated with $B$, and (3) the last element of the record in $T'$ indexed by $Y$ is `$-$'. If the above checks succeed, continue, else return an error.
        \item Suppose the entry in $T'$ indexed by $Y$ is $(e,s,Y,-)$, set $Z_1=(XP^d)^{hs}$ where $d=avf(X)$, set the last element of the entry to be $\times$, and compute the session key $k=H_{2}(Z_1,\hat{B},\hat{P})$.
        \end{enumerate}
    \end{enumerate}

\item SessionStateReveal$(s)$: $\mathcal{S}_7$ returns $\mathcal{M}$ the session key of $s$. Besides, $\mathcal{S}_7$ also returns the 3-tuple $\sigma=(Z-1,\hat{P_1},\hat{P_2})$. The reason that $\mathcal{S}_7$ can return the 3-tuple $\sigma$ of $s$ to $\mathcal{M}$ is that in the above simulation of session activations at $\hat{B}$, we compute the 3-tuple $\sigma$ of each session, and for other parties, $\mathcal{S}_7$ is also able to compute the 3-tuple $\sigma$ of each session.
\end{itemize}

\begin{proof}
As $\mathcal{S}_7$ provides $\mathcal{M}$ with the 3-tuple $\sigma$ of all exposed sessions, $\mathcal{M}$ can query $s'$ to obtain the 3-tuple $\sigma'$ of $s'$ which equals $\sigma$ of the test session $s$. This means that $\mathcal{M}$ successfully gets the 3-tuple of the test session $s$ without exposing $s$ or its matching session, namely, $\mathcal{M}$ can launch the forging attack. We have proved that if $\mathcal{M}$ can launch forging attacks, we can build a GDH solver $\mathcal{S}_6$. So $\mathcal{S}_7$ can leverage $\mathcal{S}_6$ to solve the CDH problem: $CDH(X_0,B)=\mathcal{S}_6(X_0,B)$.
$\hfill\blacksquare$
\end{proof}

\subsubsection{Analysis of Case 3}
The analysis of Case 3 is similar to Case 2, and it can be easily completed following the analysis of Case 2.

\section{Further Security Properties of $\mathsf{tpm.KE.rev}$} \label{sec:tpm.ke.rev-further}
In this section, we analyze the KCI-resistance and weak PFS (wPFS) security properties of $\mathsf{tpm.KE.rev}$.

\subsection{Resistance to KCI Attacks}
The KCI-resistance property means that even if the attacker $\mathcal{M}$ has obtained the long-term key of party $\hat{A}$, he cannot impersonate other parties to $\hat{A}$. We show that the Full UM protocol cannot achieve this security property by the following attack: $\mathcal{M}$ generates an ephemeral key pair $(y,Y=g^y)$ and initiates a session with $\hat{A}$  as the identity of $\hat{B}$; after receiving $Y$, $\hat{A}$ generates its ephemeral key pair $(x,X=g^x)$ and computes its session key $k=H_1(B^a,Y^x,\hat{A},\hat{B},X,Y)$; $\mathcal{M}$ can also compute the session key of $\hat{A}$: $k=H_1(B^a,X^y,\hat{A},\hat{B},X,Y)$.
We now prove that the MQV and SM2 key exchange protocols of $\mathsf{tpm.KE.rev}$ can achieve the KCI-resistance property.

\begin{lemma}\label{lemma:kci-rev}
Under the GDH assumption, the MQV and SM2 key exchange protocols of $\mathsf{tpm.KE.rev}$ resist KCI attacks.
\end{lemma}

\begin{proof}
Lemma \ref{lemma:kci-rev} can be easily proved by slightly modifying the proof of MQV and SM2 key exchange protocols in Section \ref{sec:analysis-rev}. The only change to the proof in Section \ref{sec:analysis-rev} is that all GDH solvers used to prove the security of MQV and SM2 key exchange protocols do not abort when $\mathcal{M}$ corrupts $\hat{A}$. The proof remains valid because the above abort operations are not used in the proof (we add the abort of GDH solvers when $\hat{A}$ is corrupted for compliance with the notion of the session exposure in the security model).
$\hfill\blacksquare$
\end{proof}

\subsection{Weak Perfect Forward Secrecy}
Hrawczyk has proved that the implicitly AKE protocols cannot achieve the full PFS property and can only achieve weak PFS \cite{HMQV05}: only the session keys that are established without the active involvement of the attacker enjoy PFS property.

\begin{lemma}\label{lemma:pfs-rev}
Under the CDH assumption, $\mathsf{tpm.KE.rev}$ provides weak PFS.
\end{lemma}

\begin{proof}
Here we only outline the idea of the proof of Lemma \ref{lemma:pfs-rev}, and the full proof can be completed following the proof in Section \ref{sec:analysis-rev}. We construct a CDH solver $\mathcal{S}_8$ if $\mathcal{M}$ successfully breaks the weak PFS security property. Let $X$ and $Y$ be the inputs to $\mathcal{S}_8$, and the goal of $\mathcal{S}_8$ is to compute $CDH(X,Y)$. $\mathcal{S}_8$ simulates the environment for $\mathcal{M}$ as follows: it sets all parties' long-term keys and chooses a random guess for the test session of the distinguishing game. We call the guessed session the g-session, denote the owner and the peer of the test session by $\hat{A}$ and $\hat{B}$ respectively, and denote their long-term private keys by $a$ and $b$, respectively. $\mathcal{S}_8$ sets the incoming and outgoing messages in the g-session to be $X$ and $Y$. If $\mathcal{M}$ does choose the g-session as the test session and wins the distinguish game, he must compute the tuple $\sigma$ of the test session in his run. We proceed to show that no matter what the type of $A$ is, $\mathcal{S}_8$ can compute $CDH(X,Y)$, i.e., $g^{xy}$.
\begin{enumerate}
\item Case $\mathsf{TPM\_ALG\_ECDH}$: $\sigma$ equals $(Z_1,Z_2,\hat{A},\hat{B},X,Y)$ where $Z_2=g^{xy}$.
\item Case $\mathsf{TPM\_ALG\_ECMQV}$: $\sigma$ equals $(Z,\hat{A},\hat{B})$ where $Z=g^{h(x+da)(y+eb)}$; since $\mathcal{S}_8$ knows $a$ and $b$, he can compute $g^{xy}$.
\item Case $\mathsf{TPM\_ALG\_SM2}$: $\sigma$ equals $(Z,\hat{A},\hat{B})$ where $Z=g^{h(a+dx)(b+ey)}$; since $\mathcal{S}_8$ knows $a$ and $b$, he can compute $g^{xy}$.
\end{enumerate}
$\hfill\blacksquare$
\end{proof}

\section{Conclusion}\label{sec:conclusion}

In this paper, we present a formal analysis of the secure communication interfaces of TPM 2.0 in a unified way. We construct a security model which takes account of the protections of the TPM on keys and protocols' computation environments and eliminates group representation attacks existing in theory by measuring the entropy of the output of the $avf()$ and $avf'()$ functions. The analysis results show that the current version of the key exchange primitive in TPM 2.0 can achieve the basic security property of modern security models, but its security is achieved under some impractical conditions: all devices in the network should be protected by the TPM, and security measures should be deployed to prevent sophisticated physical attacks, which can be used by attackers to obtain keys inside the TPM. Besides, the analysis results also show that the current version of the key exchange primitive in TPM 2.0 cannot achieve other security properties, such as KCI-resistance and weak PFS properties. We also give suggestions to engineers on how to use the key exchange primitive of TPM to implement key exchange protocols securely.

To eliminate the impractical conditions required by the TPM 2.0, we revise the key exchange primitive of TPM 2.0 and give a rigorous analysis of the revision. The analysis results show that our revision helps the key exchange primitive of TPM 2.0 not only enjoy the essential security property defined by modern AKE models in real-world networks but also achieve additional security properties: the MQV and SM2 key exchange protocols of TPM 2.0 enjoy the KCI-resistance property, and all the three protocols of TPM 2.0 enjoy the weak PFS property. Our revision only needs to modify one TPM command, and we give concrete suggestions on how to revise the TPM 2.0 specifications.

\section*{Acknowledgement}
This work was supported by the National Natural Science Foundation of China (61802375, 61602325, 61876111, 61877040), the Project of Beijing Municipal Education Commission (KM20190028005), and the Open Research Fund of State Key Laboratory of Computer Architecture, ICT, CAS (CARCH201920).

\bibliographystyle{unsrt}


\appendix
\section{Kaliski's UKS attack}\label{app:kaliski}
We describe Kaliski's UKS attack here to show how the attacker $\mathcal{M}$ successfully mounts the attack by cleverly computing its long-term private key $c$ and the ephemeral public key $X'$.

\begin{enumerate}
\item $\hat{A}$ sends an ephemeral public key $X$ to $\hat{B}$.
\item $\mathcal{M}$ intercepts $X$.
\item $\mathcal{M}$ registers to the CA a key $C=g^c$ where $c$ is cleverly computed by the following steps:
    \begin{enumerate}
    \item Choose $u\in_R Z_q$;
    \item Compute $d=avf(X)$, $X'=XA^dg^{-u}$, $e=avf(X')$, and $c=u/e$.
    \end{enumerate}
\item $\mathcal{M}$ sends $X'$ to $\hat{B}$ as the identity of $\mathcal{M}$.
\item $\mathcal{M}$ relays the ephemeral key $Y$ from $\hat{B}$ to $\hat{A}$.
\end{enumerate}

Note that $X'C^e=XA^d$, and therefore the keys computed in sessions $(\hat{A},\hat{B},X,Y)$ and $(\hat{B},\mathcal{M},Y,X')$ are identical.

\section{Xu's attacks}\label{app:xu}
Here we describe Xu's two attacks on the SM2 key exchange protocol to show that in the first attack the attacker, $\mathcal{M}$ needs to register a specific long-term public key, and in the second attack, $\mathcal{M}$ needs to use the private key of $\hat{A}$ to perform some computations.

\subsection{Attack I}
$\mathcal{M}$ selects $u\in_R Z_q$, and registers a carefully computed public key $M=Ag^u$.

\begin{enumerate}
\item $\hat{A}$ sends an ephemeral public key $X$ to $\hat{B}$.
\item $\mathcal{M}$ intercepts $X$ and sends it to $\hat{B}$ as the identity of $\mathcal{M}$.
\item $\hat{B}$ sends its ephemeral public key $Y$ to $\mathcal{M}$, and computes $Z_B=(MX^d)^{h(b+ey)}$ where $d=avf'(X)$ and $e=avf'(Y)$.
\item $\mathcal{M}$ forwards $Y$ to $\hat{A}$ as the identity of $\hat{B}$. $\hat{A}$ computes $Z_A=(BY^e)^{h(a+dx)}$.
\item $\mathcal{M}$ corrupts $Z_B$ of session $(\hat{B},\mathcal{M},Y,X)$, and then $\mathcal{M}$ can compute $Z_A$ of session $(\hat{A},\hat{B},X,Y)$: $Z_A=Z_B/(BY^e)^{hu}$, and $\mathcal{M}$ further derives the session key of $(\hat{A},\hat{B},X,Y)$.
\end{enumerate}

Note that the above attack shows that the corruption of session $(\hat{B},\mathcal{M},Y,X)$ does affect the security of session $(\hat{A},\hat{B},X,Y)$, so the SM2 key exchange protocol cannot achieve the security defined by modern AKE security models.

\subsection{Attack II}
$\mathcal{M}$ first registers a legal key $M=g^m$.

\begin{enumerate}
\item $\hat{A}$ sends an ephemeral public key $X$ to $\hat{B}$.
\item $\mathcal{M}$ intercepts $X$ and sends $X'=AX^d$ ($d=avf'(X)$) to $\hat{B}$ as the identity of $\mathcal{M}$.
\item $\hat{B}$ sends an ephemeral public key $Y$ to $\mathcal{M}$ and computes $Z_B=(MX'^{d'})^{h(b+ey)}$ where $d'=avf'(X')$ and $e=avf'(Y)$.
\item $\mathcal{M}$ forwards it to $\hat{A}$. $\hat{A}$ computes $Z_A=(BY^e)^{h(a+dx)}$ where $d=avf'(X)$ and $e=avf'(Y)$.
\item $\mathcal{M}$ corrupts $Z_B$ of session $(\hat{B},\mathcal{M},Y,X')$, computes $Z_A$ of session $(\hat{A},\hat{B},X,Y)$: $Z_A=(Z_A/(BY^e)^{hm})^{d'^{-1}}$, and further derives the session key of $(\hat{A},\hat{B},X,Y)$.
\end{enumerate}

\section{Group Representation Attack on MQV}\label{app:groupattack}
For the benefit of the reader, we present the group representation attack on MQV here. Consider such a group $G$ that the representations of its elements satisfy that the $\lceil q/2\rceil$ LSBs of the representation of points' $x$-coordinate are fixed. We use $c$ to denote the fixed value. In this case, the $Z$ value of MQV becomes $Z=g^{h(x+ca)(y+cb)}$. The attacker $\mathcal{M}$ can launch the following attack:
\begin{enumerate}
\item $\mathcal{M}$ randomly chooses $x^*\in_R Z_q$ and computes $X^*=g^{x^*}/A^c$.
\item $\mathcal{M}$ sends $X^*$ to $\hat{B}$ as the identity of $\hat{A}$.
\item $\hat{B}$ responds with $Y=g^y$, computes $Z=(X^*A^c)^{h(y+cb)}$, and computes its session key $K=H_2(Z,\hat{A},\hat{B})$.
\item $\mathcal{M}$ can also compute the session key $K=H_2((YB^c)^{hx^*},\hat{A},\hat{B})$.
\end{enumerate}

The above attack shows that $\mathcal{M}$ can impersonate $\hat{A}$ without knowing the private key of $\hat{A}$ because of the special representations of the group elements.

\end{document}